\documentclass[journal]{IEEEtran}
\usepackage{soul}
\usepackage{cite}
\usepackage{amsmath,amsfonts,amsxtra,amssymb,latexsym,amscd,amsthm,mathrsfs,bm}
\usepackage{algorithmic}
\usepackage{algorithm}
\usepackage{graphics}
\usepackage{graphicx}
\usepackage{hyperref}
\usepackage{textcomp}
\usepackage{balance}
\usepackage{booktabs}
\usepackage{caption}
\usepackage{subcaption}
\usepackage{cuted, nccmath}
\usepackage{amsmath}
\usepackage{lipsum}
\usepackage[dvipsnames]{xcolor}
\usepackage[utf8]{inputenc}
\usepackage[english]{babel}
\usepackage{filecontents}
\usepackage{xcolor}
\usepackage{thmtools}

\usepackage{stackengine}
\usepackage{stfloats}% <-- added

\usepackage{mathtools}
\DeclarePairedDelimiter\abs{\lvert}{\rvert}

\makeatletter
\renewcommand{\fnum@figure}{Fig. \thefigure}
\makeatother

\newtheorem{lemma}{Lemma}

\newtheoremstyle{mythm}%
{3pt}% Space above
{3pt}% Space below
{}% Body font 
{}% Indent amount
{\bfseries}% ⟨Theorem head font⟩
{}% ⟨Punctuation after theorem head ⟩
{.5em}% Space after theorem head 
{}%

\declaretheoremstyle[
  headpunct=\textup{:},
  bodyfont=\normalfont,
]{myremark}
\theoremstyle{myremark}
\newtheorem{remarkex}{Remark}
 
\newenvironment{remark}
  {\pushQED{\qed}\remarkex}
  {\popQED\endremarkex}

\ifCLASSINFOpdf
\else
\fi
\usepackage[compact]{titlesec}  
\titlespacing{\section}{0pt}{1pt}{1pt}
\titlespacing{\subsection}{0pt}{1pt}{1pt}
\titlespacing{\subsubsection}{0pt}{1pt}{1pt}
\titleformat{\subsubsection}[runin]{\normalfont\itshape}{\arabic{subsubsection}.}{0pt}{}[:]
\begin{document}
	%
	% paper title
	% Titles are generally capitalized except for words such as a, an, and, as,
	% at, but, by, for, in, nor, of, on, or, the, to, and up, which are usually
	% not capitalized unless they are the first or last word of the title.
	% Linebreaks \\ can be used to get better formatting as desired.
	% Do not put math or special symbols in the title.
	\title{\huge Jointly Optimizing Power Allocation and Device Association for Robust IoT Networks under Infeasible Circumstances}
	% Exploiting Standard Interference Function
	%
	% author names and IEEE memberships
	% note positions of commas and nonbreaking spaces ( ~ ) LaTeX will not break
	% a structure at a ~ so this keeps an author's name from being broken across
	% two lines.
	% use \thanks{} to gain access to the first footnote area
	% a separate \thanks must be used for each paragraph as LaTeX2e's \thanks
	% was not built to handle multiple paragraphs
	%
	
	\author{Nguyen~Xuan~Tung, Trinh~Van~Chien, \textit{Member}, Dinh~Thai~Hoang, \textit{Senior Member}, \textit{IEEE}, Won~Joo~Hwang, \textit{Senior Member}, \textit{IEEE}
		%        and~Jane~Doe,~\IEEEmembership{Life~Fellow,~IEEE}% <-this % stops a space
		%\thanks{This paper is supported by ...}% <-this % stops a space
		\thanks{Nguyen Xuan Tung is with the Department of Information Convergence Engineering, Pusan National University, Busan 46241, Republic of Korea (email: {tung.nguyenxuan1310@pusan.ac.kr}). 
  
        Trinh Van Chien is with the School of Information and Communication Technology (SoICT), Hanoi University of Science and Technology (HUST), Vietnam (email: chientv@soict.hust.edu.vn). 
        
        Dinh Thai Hoang is with the School of Electrical and Data Engineering, University of Technology Sydney, Sydney, NSW 2007, Australia (e-mail: hoang.dinh@uts.edu.au). 
        
        Won-Joo Hwang is with the School of Computer Science and Engineering, Center for Artificial Intelligence Research, Pusan National University, Busan 46241, South Korea (e-mail: wjhwang@pusan.ac.kr).}
		%	and the Department of Convergence Engineering for Intelligent Drone, Sejong
		%	University, Seoul 05006, South Korea  (e-mail: cheol.jeong@ieee.org).}% <-this % stops a space
		%\thanks{Manuscript received April 19, 2005; revised August 26, 2015.}
	}
	
	% note the % following the last \IEEEmembership and also \thanks - 
	% These prevent an unwanted space from occurring between the last author-name
	% and the end of the author line. i.e., if you had this:
	% 
	% \author{....lastname \thanks{...} \thanks{...} }
	%                     ^------------^------------^----Do not want these spaces!
	%
	% a space would be appended to the last name and could cause every name on that
	% line to be shifted left slightly. This is one of those "LaTeX things". For
	% instance, "\textbf{A} \textbf{B}" will typeset as "A B" not "AB". To get
	% "AB" then you have to do: "\textbf{A}\textbf{B}"
	% \thanks is no different in this regard, so shield the last } of each \thanks
	% that ends a line with a % and do not let a space in before the next \thanks.
	% Spaces after \IEEEmembership other than the last one are OK (and needed) as
	% you are supposed to have spaces between the names. For what it is worth,
	% this is a minor point as most people would not even notice if the said evil
	% space somehow managed to creep in.

	% The paper headers
	\markboth{}%Journal of \LaTeX\ Class Files,~Vol.~14, No.~8, August~2015
	{Shell \MakeLowercase{\textit{et al.}}: Jointly Optimizing Power Allocation and Device Association for Robust IoT Networks under Infeasible Circumstances}
	%  Exploiting Standard Interference Function
	% The only time the second header will appear is for the odd numbered pages
	% after the title page when using the twoside option.
	% 
	% *** Note that you probably will NOT want to include the author's ***
	% *** name in the headers of peer review papers.                   ***
	% You can use \ifCLASSOPTIONpeerreview for conditional compilation here if
	% you desire.

	% If you want to put a publisher's ID mark on the page, you can do it like
	% this:
	%\IEEEpubid{0000--0000/00\$00.00~\copyright~2015 IEEE}
	% Remember, if you use this you must call \IEEEpubidadjcol in the second
	% column for its text to clear the IEEEpubid mark.

	% use for special paper notices
	%\IEEEspecialpapernotice{(Invited Paper)}

	% make the title area
	\maketitle
	
	% As a general rule, do not put math, special symbols, or citations
	% in the abstract or keywords.
    
	\begin{abstract}
        Jointly optimizing power allocation and device association is crucial in Internet-of-Things (IoT) networks to ensure devices achieve their data throughput requirements. Device association, which assigns IoT devices to specific access points (APs), critically impacts resource allocation. Many existing works often assume all data throughput requirements are satisfied, which is impractical given resource limitations and diverse demands. When requirements cannot be met, the system becomes infeasible, causing congestion and degraded performance. To address this problem, we propose a novel framework to enhance IoT system robustness by solving two problems, comprising maximizing the number of satisfied IoT devices and jointly maximizing both the number of satisfied devices and total network throughput. These objectives often conflict under infeasible circumstances, necessitating a careful balance. We thus propose a modified branch-and-bound (BB)-based method to solve the first problem. An iterative algorithm is proposed for the second problem that gradually increases the number of satisfied IoT devices and improves the total network throughput. We employ a logarithmic approximation for a lower bound on data throughput and design a fixed-point algorithm for power allocation, followed by a coalition game-based method for device association. Numerical results demonstrate the efficiency of the proposed algorithm, serving fewer devices than the BB-based method but with faster running time and higher total throughput.
	\end{abstract} 
	% Note that keywords are not normally used for peer-review papers.
	\begin{IEEEkeywords}
		Dual-objective optimization, IoT, service management, power allocation, device association.
	\end{IEEEkeywords}
	
	% For peer review papers, you can put extra information on the cover
	% page as needed:
	% \ifCLASSOPTIONpeerreview
	% \begin{center} \bfseries EDICS Category: 3-BBND \end{center}
	% \fi
	%
	% For peer review papers, this IEEEtran command inserts a page break and
	% creates the second title. It will be ignored for other modes.
	\IEEEpeerreviewmaketitle
    \section{Introduction}
    The Internet of Things (IoT) has transformed how we interact with the world, connecting billions of devices to the Internet and enabling seamless information exchange and automation across sectors such as healthcare, agriculture, transportation, and smart cities \cite{Lin2017}. This progress, driven by advances in sensor technology, wireless communication, cloud computing, and data analytics, has made IoT networks integral to modern infrastructure. However, as IoT networks expand, managing the network to maintain a high quality of service poses several challenges \cite{Ouedraogo2022, Sisinni2018}, especially interference management among connected devices, which weakens wireless signals, compromises data transmission, and leads to unreliable connections. Therefore, effective interference management \cite{Zhao2016} is crucial for maintaining signal quality, reducing packet loss, minimizing latency, and ensuring reliable communication, particularly in critical applications like healthcare monitoring \cite{Catarinucci2015}, industrial automation \cite{Xu2014}, and smart grid systems \cite{Lin2017}. IoT networks, typically composed of diverse devices utilizing various wireless technologies and communication protocols \cite{Raza2017}, require robust interference management techniques to ensure smooth and seamless operation \cite{Ozcan2020, Gandotra2018, Gbadamosi2022}. Moreover, interference increases energy consumption as devices contend with competing signals, making advanced power control and interference avoidance mechanisms essential for extending networks' lifespan \cite{Yu2023, Liu2019Interference}.

    A large number of researches has been conducted to manage interference and enhance IoT networks performance \cite{Zhao2016, Yates1995, Park2021, Liu2019_NOMA_Deep_Cognitive_Perspective, Gbadamosi2022}. Approaches such as non-orthogonal multiple access (NOMA) \cite{Park2021, Liu2019_NOMA_Deep_Cognitive_Perspective} have focused on optimizing data throughput but rely on successive interference cancellation (SIC), which demands significant energy and computational resources, making it less practical for many IoT scenarios. Alternatively, a fixed-point power control algorithm was proposed in    \cite{Yates1995} to find the minimum power allocation, enabling users to overcome interference and ensure an acceptable connection. Despite these advances, many studies assumed sufficient resources to meet the data rate requirements for all IoT devices simultaneously \cite{Shen2018_PractionalProgramHandling, Pang2008Empty_nonEmpty_game, Obiedollah2020_FeasibilityCheck}. This assumption is often invalid in practice due to power and data rate constraints. When the system cannot meet these requirements, we define this circumstance as ``infeasible". Once infeasible issues happen, existing approaches designed for feasible conditions become inapplicable, rendering the system unresponsive. This can lead to congestion, transmission delays, and degraded performance \cite{JAIN2022100678}. Moreover, IoT networks often face restricted resources and intense interference \cite{Lin2017}, which increases the likelihood of encountering infeasible conditions. Addressing the infeasible issue is crucial to ensuring the practical viability and scalability of IoT deployments, especially for critical applications like healthcare and traffic management \cite{Fuqaha2015}. For example, at a parking place with IoT devices like cameras and parking sensors, encountering infeasible circumstances could prevent the system from serving users, causing disruptions in traffic management and safety operations. To avoid such disruptions, it is essential to strategically address infeasibility by serving a subset of IoT devices to maintain critical services.
    % In a different approach, \cite{Gbadamosi2022}, the authors proposed a channel assignment mechanism for an IoT system where device-to-device (D2D) communication reuses narrowband IoT (NB-IoT) bandwidth, managing interference while ensuring minimum service for both user types.
    
    This paper introduces an innovative framework aimed at enhancing the robustness of IoT systems under infeasible circumstances. Specifically, we examine an IoT network where multiple APs communicate with IoT devices in the downlink utilizing the same frequency resource. Each AP transmits signals to multiple IoT devices, while an IoT device receives data from one AP only. We assume perfect channel state information (CSI) is available, similar to many other related works in the literature such as \cite{Gbadamosi2022, Pang2008Empty_nonEmpty_game, Obiedollah2020_FeasibilityCheck, Stridh2006Interference_CongestionControl, Antonioli2019IBC, Yu2012CoordinatedScheduling, Chien2020CongestionICC, Chien2021ScatteringChannelICC, van2021uplink, Chi2005}. This assumption, while idealistic, is essential for developing a robust framework addressing the infeasible issue. By first establishing a baseline performance under these ideal conditions, we lay the groundwork for future research that will relax these assumptions and explore more realistic scenarios, including imperfect CSI. The system information is processed centrally to manage resources effectively, with an IoT device's QoS requirement considered satisfied when its data throughput exceeds its demand. A critical aspect of this framework is the joint optimization of AP-IoT device associations and power allocation. Properly associating IoT devices with APs and allocating power resources efficiently are keys to optimizing signal strength, coverage, and interference management \cite{Jamali9542950}. However, to the best of our knowledge, there was no literature investigating the effect of AP-IoT device association under infeasible circumstances.

    % When the number of APs and their power budget is limited, meeting the QoS requirements for many IoT devices becomes challenging, making infeasibility inevitable.
    
    % Therefore, the system falling into infeasible circumstances is inevitable. To prevent the system from becoming unresponsive and to enhance its practicality, it is crucial to identify and provide services to as many IoT devices as the system can satisfy, addressing the remaining devices sequentially.
    
    In particular, we divide the IoT devices into two sets. The first set comprises IoT devices that the system can satisfy the QoS requirement, and the second set consists of the remaining IoT devices. We aim to maximize the number of IoT devices in the satisfied set and serve them first, with the remaining devices being served later. By doing so, the system handles potentially infeasible situations more effectively, maintaining system functionality and efficient resource allocation. We then consider two optimization problems for power allocation and AP-IoT device association: the first is to maximize the number of satisfied IoT devices, and the second is a dual-objective optimization problem to optimize both the number of satisfied IoT devices and the total network throughput. The first problem is solved by utilizing a modified BB algorithm. This method provides a straightforward approach for maximizing the number of satisfied IoT devices while also offering a simple mechanism to assess system infeasibility. For the dual-objective problem, we recognize a conflict between the objectives: increasing the number of satisfied devices can reduce overall network throughput due to stretched resources, while prioritizing network throughput may decrease the number of satisfied devices. To resolve this problem, our approach prioritizes optimizing the number of satisfied IoT devices. We propose an iterative algorithm to find and add IoT devices into the satisfied set. This is achieved by alternately optimizing power allocation and AP-IoT device association. The method effectively increases the number of IoT devices meeting their QoS requirements by reallocating resources strategically, ensuring that more devices are served under the system's constraints. The non-convex structure of the data rate formula is managed using the log-approximation method \cite{Singh7000604}. Extensive simulations demonstrate that the modified BB algorithm can serve more IoT devices but with a lower total network throughput and a higher running time. These limitations motivate future work to improve the algorithm's computational efficiency and scalability for better performance in large-scale IoT networks. Besides, compared with the equal power allocation scheme, the proposed power allocation algorithm can improve QoS satisfaction by approximately $50$\%. Our main contributions can be summarized as follows:
    \begin{itemize}
        \item We propose a novel framework to address the challenges of infeasible circumstances in IoT networks where it is impossible to satisfy the requirements of all IoT devices simultaneously. Our framework introduces two optimization problems: maximizing the number of satisfied IoT devices and a dual-objective problem that optimizes both device satisfaction and total network throughput. This approach ensures efficient operation under challenging conditions, optimizing both individual device satisfaction and overall network performance.
        \item We develop a branch-and-bound-based algorithm to address the number of satisfied IoT device maximization problems and identify infeasible circumstances. For the dual-objective problem, we proposed an iterative algorithm that increases the number of satisfied devices by solving total data throughput maximization problems with respect to the power allocation and AP-IoT device association. We introduce a low-complexity dual fixed-point algorithm based on the Karush-Kuhn-Tucker conditions for power allocation and a coalition game model to determine the optimal device association strategy. 
        \item We conduct extensive simulations and benchmark comparisons to validate the efficiency of the proposed framework and provide deeper insights. The results reveal that branch-and-bound can provide services to more IoT devices when the requirements from IoT devices are not too large, albeit with higher running time. For the dual-objective problem, the proposed AP-IoT device association outperforms the geometrical-based AP association, which assigns each IoT device to the closest AP. Additionally, the proposed power allocation surpasses the fixed power allocation, where equal power is distributed to all devices.
    \end{itemize}
	
	\emph{Paper Organization:} The rest of the paper is organized as follows. Related works are reviewed in \ref{Section: Related Work}. In Section~\ref{SystemModel}, we describe the system model and formulate the optimization problem with respect to the power allocation and the AP-IoT device association. In Section \ref{OptimiaztionApproach}, we propose an efficient algorithm that obtains the solution to the considered optimization problem by exploiting the Lagrangian and standard interference function. Numerical results are obtained and discussed in Section \ref{ExperimentalEvaluation}. Finally, Section \ref{Conclusion} concludes the paper.

    \emph{Notation:} We use boldface lowercase letters and boldface uppercase letters to denote vectors and matrices, respectively. Let $\mathbf{a}^T$ denote the transpose of vector $\boldsymbol{a}$. We define the circularly symmetric complex Gaussian distribution with zero mean and variance $\sigma^2$ by $\mathcal{CN}(0,\sigma^2)$.
	% . We define $[1:n] = \{1, 2,..., n\}$
    % The component-wise inequality between vectors $\boldsymbol{a}$ and $\boldsymbol{b}$ is denoted as $\boldsymbol{a}\succeq\boldsymbol{b}$. The all-one vector with $N$ elements is denoted by $\boldsymbol{1}_N$. 
 
    \section{Related Work}\label{Section: Related Work}
    Despite the importance of the infeasible issue in IoT networks, research on this challenge is still limited. Some studies, such as those referenced in \cite{Pang2008Empty_nonEmpty_game, Obiedollah2020_FeasibilityCheck}, discussed the potential for infeasibility and provided conditions to ensure the system avoids such scenarios. Therefore, these works primarily focus on preventing infeasibility rather than offering solutions when the system becomes infeasible. The paper \cite{Stridh2006Interference_CongestionControl} addressed the infeasibility issue in maximizing the total network rate under minimum data rate constraints but assumed an infinite power budget, which is impractical for real-world applications. Indeed, when a finite power budget is considered, the system becomes more complex, and the proposed approach is not suitable anymore. 

    Further studies have addressed infeasible problems under both data rate constraints and power limitations. For instance, \cite{Chi2005} suggested adjusting the data throughput of users based on channel gain and interference levels to cope with infeasibility. However, lowering the QoS requirement is unsuitable for non-delay-tolerant or precise data requirement applications. Similarly, the papers \cite{Antonioli2019IBC} and \cite{Yu2012CoordinatedScheduling} also tackled the problem of infeasibility under limited power budgets and minimum data throughput requirements. These studies addressed infeasibility by using a BB method to select a subset of users, ensuring they meet minimum data rates while optimizing overall system performance, but these approaches prioritize system performance over maximizing user satisfaction. In \cite{Chien2020CongestionICC, Chien2021ScatteringChannelICC, van2021uplink}, the authors addressed infeasibility in uplink cell-free networks with power and data throughput constraints. Specifically, \cite{Chien2020CongestionICC} proposed a fixed-point algorithm to allocate power resources to users based on the data throughput constraints. The maximum power is assigned to users who cannot meet the required data throughput. The works \cite{Chien2021ScatteringChannelICC, van2021uplink} proposed a new update for users who cannot meet the required data throughput to mitigate interference. Nevertheless, similar to \cite{Antonioli2019IBC} and \cite{Yu2012CoordinatedScheduling}, these studies focus on optimizing overall objectives rather than maximizing the number of satisfied users.
    
    In contrast, our approach takes a different perspective. We prioritize maximizing the number of satisfied IoT devices, even at the cost of lower overall throughput. This strategy enhances reliability and ensures that more devices can operate correctly, which is crucial for applications where individual device performance is critical. Additionally, while the joint AP-IoT device association and power allocation has been recognized for their potential to enhance overall system performance \cite{Jamali9542950, Sun7070670, Qian6476076}, existing studies have not explored their impact under infeasible circumstances. Our framework optimizes both power allocation and AP-IoT device association, directly addressing infeasibility to maintain system functionality even under constrained resources, thereby filling the gaps identified in previous research.

    % The total network throughput maximization problem has been widely considered in various communication networks \cite{Papandriopoulos2006, Ju2014, Kai2019, Chien2020}. To handle the difference of convex structure, the authors in \cite{Papandriopoulos2006} utilized a lower bound on the logarithm formula of the data throughput. However, optimizing the system solely by the power budget, \cite{Park2021, Shi2011WMMSE}, leads to disparities in service quality, while some IoT devices experience an unnecessarily high data throughput, others suffer from very poor service. When all users request minimum data throughput at the same time, the system may have no solution, namely infeasible circumstances. Indeed, the research on feasible scenarios, \cite{Shen2018_PractionalProgramHandling, Pang2008Empty_nonEmpty_game, Obiedollah2020_FeasibilityCheck}, is prevalent, while infeasible cases are somewhat less popular due to the complexity of empty feasible domains.

	\section{System Models and Optimization Problem Formulation}\label{SystemModel}
 
	We consider an IoT system in which $K$ APs communicate with $N$ IoT devices in the downlink transmission as illustrated in Fig.~\ref{SystemModelFigure}. Similar to the concept of other NB-IoT networks, all APs and IoT devices are located in a small area due to the limitation in communication range and the power budget. APs and IoT devices are equipped with a single antenna. All communication links between APs and IoT devices share the same time and frequency resource. Therefore, interference appears in this system, and the IoT devices apply the maximum likelihood estimator to detect the desired signal. The NOMA is not suitable for this system since the IoT devices' resources are not sufficient to handle the SIC technique \cite{Estrello22_SIC}. We respectively denote the set of IoT devices and APs by
     \begin{equation}
       \mathcal{N}\triangleq \{1, \ldots,N\} \mbox{ and } \mathcal{K}\triangleq \{1, \ldots,K\}.
    \end{equation}
    In the considered network, a centralized server is utilized to calculate the resource allocation based on the gathered information consisting of the power budget of APs and the CSI between APs and IoT devices. All APs are centrally managed and controlled by a server. Whenever there are changes within the network, the APs promptly relay this information back to the server and await further instructions. As we consider the downlink communications in an IoT network, we assume an AP can transmit data to multiple IoT devices simultaneously. However, one IoT device can only be communicated with one AP due to IoT devices' hardware constraints. 
    \begin{figure}[t]
    		\includegraphics[trim=0.2cm 0.4cm 0cm 0cm, clip=true, width=3.5in]{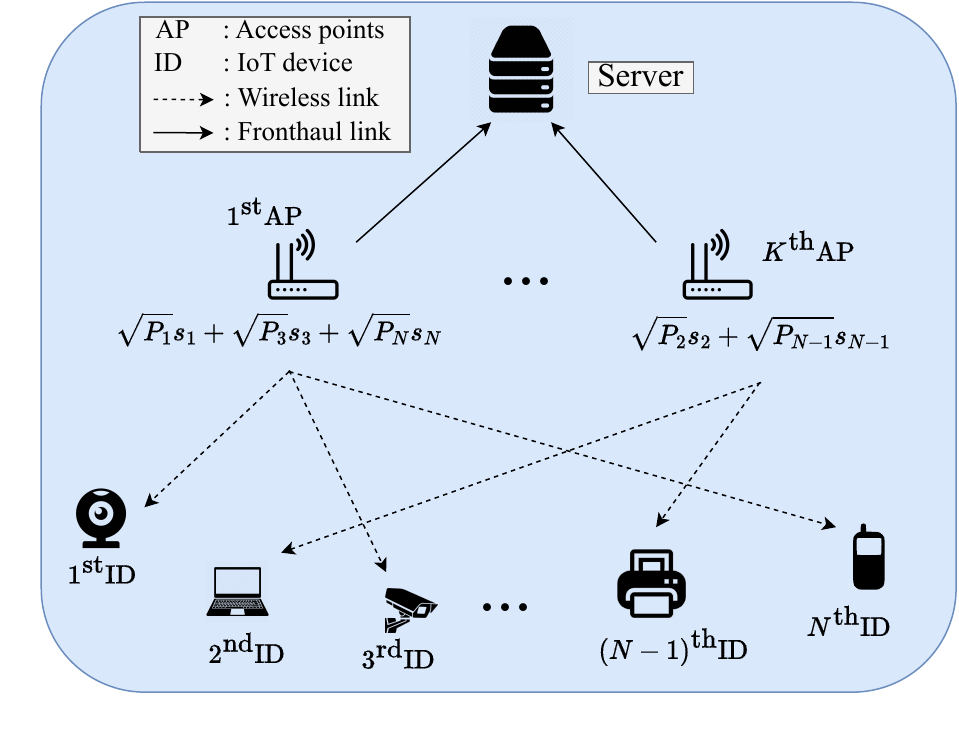}
    		\caption{The considered model of a downlink IoT network with the participation of multiple APs and IoT devices.}
    		\label{SystemModelFigure}
            \vspace{-2.0em}
    \end{figure} 
    \subsection{Signal Model}
    During the downlink data transmission, the $n$-th IoT device receives the desired signal $s_n \in \mathbb{C}$ with $|s_n|^2 =1$ from its associated AP. The allocated power for the $n$-th IoT device is constrained by
     \begin{equation} \label{eq:PowerCst}
     0\leq \sum_{n\in \mathcal{N}} \mu_{k,n}P_{n}\leq  P_{k}^{\textrm{max}},
     \end{equation}
     where $P_{k}^{\textrm{max}}$ is the maximum power that the associated AP can allocate to the signal. We define $\mathbf{p}$ as the vector of power allocation of all IoT devices, $\mathbf{p}=[P_1,\ldots,P_N]^T \in \mathbb{R}^{N}$. The binary variable $\mu_{k,n}$ represents the AP-IoT device association. Specifically, in \eqref{eq:PowerCst}, if $\mu_{k,n}=1$, the desired signal of the $n$-th IoT device is transmitted from the $k$-th AP. Otherwise, $\mu_{k,n}=0$ means no communication between the $n$-th IoT device and the $k$-th AP. The received signal at the $n$-th IoT device, denoted by $y_n \in \mathbb{C}$ is given as
    \begin{multline}\label{ReceiveSN}
        y_{n}=\underset{\textrm{Desired signal}}{\underbrace{\mu_{k,n} h_{k,n} \sqrt{P_{n}} s_{n}}}+ \\ \underset{\textrm{Mutual interference}}{\underbrace{\sum_{n' \neq n, n\in \mathcal{N}}\sum_{k\in \mathcal{K}} \mu_{k,n'} h_{k,n} \sqrt{P_{n'}} s_{n'}}} + \underbrace{n_{n}}_{\textrm{Noise}},
    \end{multline}
    where $h_{k,n}$ denotes the propagation channel between the $n$-th IoT device and the $k$-th AP. We define $\boldsymbol{h}_{n}$ as the vector of the channel between all APs and the $n$-th AP, $\boldsymbol{h}_{n}=[h_{1,n},\ldots,h_{K,n}]^T$. In the right-hand side of \eqref{ReceiveSN}, the first term represents the desired signal for the $n$-th IoT device. The second term denotes mutual interference due to the reusing of frequency resources. The last term $n_{n}$ is the additive white Gaussian noise (AWGN) at the $n$-th IoT device, which is distributed as $\mathcal{CN}(0,\sigma^2)$, where the variance is calculated as $\sigma^2=B N_0$ with $N_0$ and $B$ being the power spectral density of the noise and the bandwidth of each IoT device, respectively. Since each IoT device is associated with only one AP, the following constraints are fulfilled,
    \begin{equation}
    \sum_{k\in \mathcal{K}}\mu_{k,n}=1, \forall n\in \mathcal{N}.
    \end{equation}
    The AP-IoT devices association between all APs and the $n$-th AP is denoted by $\boldsymbol{\mu}_n=[\mu_{1,n},\ldots,\mu_{K,n}]^T \in \mathbb{B}^K$, and we have $\boldsymbol{\Psi}=[\boldsymbol{\mu}_1,\ldots,\boldsymbol{\mu}_N] \in \mathbb{B}^{K \times N}$. Note that $\pmb{\Psi}$ and $\mathbf{p}$ contain the optimization variables handled hereafter. Since all the APs transmit their signals to the IoT devices at the same time on the same frequency,  the signal-to-interference-plus-noise (SINR) ratio for the $n$-th IoT device is computed as
    \begin{equation}\label{SINR}
        \begin{aligned}
    		\gamma_{n}(\boldsymbol{\Psi},\textbf{\textrm{p}})=&\frac{ \sum\limits_{k\in\mathcal{K}}|\mu_{k,n}h_{k,n}|^2 P_{n}}{\sum\limits_{n'\neq n, n'\in\mathcal{N}}\sum\limits_{k\in\mathcal{K}}|\mu_{k,n'}h_{k,n}|^2 P_{n'}+\sigma^2}\\
            =&\frac{ |\boldsymbol{\mu}_{n}^T \boldsymbol{h}_{n}|^2 P_{n}}{\underset{{n'\neq n,n'\in \mathcal{N}}}{\sum}|\boldsymbol{\mu}_{n'}^T\boldsymbol{h}_{n}|^2 P_{n'}+\sigma^2},
        \end{aligned} 
    \end{equation} 
    where the numerator is come from the fact that   $|\mu_{1,n}h_{1,n}+...+\mu_{K,n}h_{K,n}|=|\boldsymbol{\mu}_{n}^T \boldsymbol{h}_n|=|h_{k,n}|$ if $\mu_{k,n}=1$ and $\mu_{k',n}=0, \forall k'\neq k$. Similarly, in the denominator, we have $\sum\limits_{k\in\mathcal{K}}|\mu_{k,n'}h_{k,n}|^2 P_{n'} = |\boldsymbol{\mu}_{n'}^T\boldsymbol{h}_{n}|^2 P_{n'}$.
    The data throughput of the $n$-th IoT device is then computed as
    	\begin{equation}\label{RateCalculation}
    		R_{n}(\boldsymbol{\Psi},\textbf{\textrm{p}})=\log_2 \left(1+\gamma_{n}(\boldsymbol{\Psi},\textbf{\textrm{p}})\right).
    	\end{equation}
    After that, the total network throughput can be determined as follows:
    \begin{equation}\label{TotalRate}
    		\begin{aligned}
    			R_{\textrm{tot}}(\boldsymbol{\Psi},\textbf{\textrm{p}})=\sum_{n\in \mathcal{N}}R_n(\boldsymbol{\Psi},\textbf{\textrm{p}}). % , \mbox{ [Mbps]}
    		\end{aligned}
    \end{equation}
     
    \begin{remark}\label{remark1}
    It is important to investigate the data throughput of each IoT device under the optimal AP-IoT device association that is of practical interest. Accordingly, the system can optimize radio resources such as power budget or allocating frequency. The IoT network model considered in this paper can vary based on specific applications, for example, in traffic management \cite{Ning2019}, relay systems \cite{Lyu2020, Singh2015}, or the Tactile Internet \cite{Budhiraja2019}, where IoT devices request the APs to send information as quickly as possible for timely responses to the environment. 
    \end{remark}
    \subsection{Problem Formulation} 
    \subsubsection{Conventional joint power allocation and device association optimization problem}
    An important key task of every IoT network is to maximize the overall network throughput while guaranteeing that every IoT device can meet its required data throughput. Due to the limited power budget, the restriction in the association between APs and IoT devices and the interference between IoT devices, the power allocation and the AP-IoT devices association have significant impacts on optimizing the system performance. Moreover, optimizing the power allocation depends on the given AP-IoT devices association, and vice versa. Therefore, jointly optimizing both the power allocation and the AP-IoT devices association is needed. The total network throughput optimization for a downlink communication IoT system is then formulated as:
    \begin{subequations}\label{Problemv1}
        \begin{align}
            && \underset{\substack{ \boldsymbol{\Psi}, \textbf{\textrm{p}} }}{\textrm{maximize}} & \quad R_{\textrm{tot}}(\boldsymbol{\Psi},\textbf{\textrm{p}}), \\
            &&\textrm{subject to}\hspace{0.15cm}&\quad\underset{n\in \mathcal{N}}{\sum} \mu_{k,n}P_{n}\leq  P_{k}^{\textrm{max}},  \forall k\in \mathcal{K},\label{V1b} \\
            &&&\quad R_n(\boldsymbol{\Psi},\mathbf{p})\geq R_n^{\textrm{thr}}, \forall n\in \mathcal{N},\label{V1c} \\
            &&& \quad\sum_{k\in \mathcal{K}}\mu_{k,n}=1, \forall n\in \mathcal{N}, \label{V1d} 
        \end{align} 
    \end{subequations} 
    where $R_n^{\textrm{thr}}$~[Mbps] is data throughput required by the $n$-th IoT device. The constraints \eqref{V1b} represent the limited transmit power of each AP. The constraints \eqref{V1d} are to guarantee that each IoT device is associated with only one AP. The above problem is non-convex due to the non-convex form of the data throughput of IoT devices at the objective function and the required data throughput constraints. Note that the data throughput of every IoT device is reduced significantly because of the mutual interference, especially with large-size networks. Moreover, the power budget of each AP is limited, therefore, the constraints in \eqref{V1c} are difficult to be satisfied and the problem in \eqref{Problemv1} can fall into an infeasible circumstance.
    
    In case the AP power is too limited, the network throughput maximization without rate constraints can be considered
    \begin{subequations}\label{Problemv1-2}
        \begin{align}
            && \underset{\substack{ \boldsymbol{\Psi}, \textbf{\textrm{p}} }}{\textrm{maximize}} & \quad R_{\textrm{tot}}(\boldsymbol{\Psi},\textbf{\textrm{p}}) ,\\
            &&\textrm{subject to}\hspace{0.15cm}&\quad\underset{n\in \mathcal{N}}{\sum} \mu_{k,n}P_{n}\leq  P_{k}^{\textrm{max}},  \forall k\in \mathcal{K}, \\
            &&& \quad\sum_{k\in \mathcal{K}}\mu_{k,n}=1, \forall n\in \mathcal{N},
        \end{align}
    \end{subequations} 
    However, in this case, the system tends to allocate most of the power resources to the IoT devices with the best channel gains to ultimately maximize the overall throughput. Consequently, the remaining IoT devices experience very low QoS or even receive no service from APs. 
    % \vspace{-10em}
    % In harsh propagation environments, only a subset of IoT devices may meet their required data throughput with unnecessarily high data throughput, while the other IoT devices could not meet the service requirements. 
    \subsubsection{The number of satisfied IoT devices maximization problem and the dual-objective problem for the power allocation and the AP association under infeasible circumstances}
    As discussed in the previous subsection, the conventional problems in \eqref{Problemv1} and \eqref{Problemv1-2} have coherent weaknesses. The key issue is the difficulty of satisfying the individual requirements for every IoT device. Therefore, it is better to identify and serve a set of IoT devices in the system's capability only. Accordingly, we propose to divide IoT devices into two sets. The first set contains IoT devices meeting all the constraints, and the remaining IoT devices belong to the second set. To increase the number of satisfied IoT devices in the system, we formulate the following problem: 
    \begin{subequations}\label{ProblemMaxQ}
        \begin{align}
            && \underset{\substack{ \boldsymbol{\Psi}, \textbf{\textrm{p}} }}{\textrm{maximize}} & \quad |\mathcal{Q}(\boldsymbol{\Psi},\textbf{\textrm{p}})| ,\\
            &&\textrm{subject to}\hspace{0.15cm}&\quad\underset{n\in \mathcal{N}}{\sum} \mu_{k,n}P_{n}\leq  P_{k}^{\textrm{max}}, \forall k\in \mathcal{K}, \label{VQ1b}\\
            &&&\quad R_n(\boldsymbol{\Psi},\mathbf{p})\geq R_n^{\textrm{thr}}, \forall n\in \mathcal{Q}(\boldsymbol{\Psi},\textbf{\textrm{p}}),\label{VQ1c} \\
            &&& \quad\sum_{k\in \mathcal{K}}\mu_{k,n}=1, \forall n\in \mathcal{N},\label{VQ1d} 
        \end{align}
    \end{subequations}
    where $\mathcal{Q}(\boldsymbol{\Psi},\textbf{\textrm{p}})$ is the set of satisfied IoT devices defined as
    \begin{equation}\label{QsetDeterminantEq}
        \mathcal{Q}(\boldsymbol{\Psi},\textbf{\textrm{p}})=\{n| n\in \mathcal{N}, R_n(\boldsymbol{\Psi},\textbf{\textrm{p}})\geq R_n^{\textrm{thr}}\}.
    \end{equation}
    The problem in \eqref{ProblemMaxQ} is a combinatorial problem with the mixed discrete and continuous variables of the AP-IoT devices association and the power allocation. There is a total of $\sum\limits_{i=1}^N C^i_N$\footnote{$C^i_N = \frac{N!}{i! (N-i)!}$ is the number of $i$-combinations from the set of $N$ elements. Here, $N! = \prod^{N-1}_{j=0}(N-j).$} candidates of selecting IoT devices for the satisfied set and $N^K$ candidates of AP-IoT devices association, running the exhaustive search is inapplicable, especially with large network sizes. Therefore, finding the optimal solution for the problem in \eqref{ProblemMaxQ} is impractical. We can change the objective function of \eqref{ProblemMaxQ} to maximize the minimum data throughput of all IoT devices, namely the max-min fairness \cite{Chien2020}. However, the max-min fairness problem in large-scale networks as IoT networks usually leads to a zero rate. Meanwhile, the problem in \eqref{ProblemMaxQ} ensures the minimum data throughput of an IoT device is greater than the request. The feasible set of the problem in \eqref{ProblemMaxQ} is a nonempty set as claimed in Lemma \ref{Lemma:QsetSolution}.
    \begin{lemma}\label{Lemma:QsetSolution}
        For a given number of APs and the limited power budget at each AP, there always exists a solution to  problem~\eqref{ProblemMaxQ}.
    \end{lemma}
    \begin{proof}
        For a given feasible point $(\boldsymbol{\Psi}', \textbf{\textrm{p}}')$ satisfying both \eqref{VQ1b} and \eqref{VQ1d}, the satisfied IoT devices set is determined based on \eqref{QsetDeterminantEq}, $\mathcal{Q}(\boldsymbol{\Psi}',\textbf{\textrm{p}}')$. In the worst case, $\mathcal{Q}(\boldsymbol{\Psi}',\textbf{\textrm{p}}')=\emptyset$, and thus the optimal value of \eqref{ProblemMaxQ} is $|\mathcal{Q}(\boldsymbol{\Psi}',\textbf{\textrm{p}}')|=|\emptyset|=0$. In contrast, $0< |\mathcal{Q}(\boldsymbol{\Psi}',\textbf{\textrm{p}}')| < K$ indicates that a subset of IoT devices will be served by at least their quality of service requirements. If $|\mathcal{Q}(\boldsymbol{\Psi}',\textbf{\textrm{p}}')| = K$, the system can guarantee individual services to all the IoT devices. Note that the feasible set of problem \eqref{ProblemMaxQ} depends on the constraints \eqref{VQ1b} and \eqref{VQ1d}. As long as the system exists at least one AP and the power budget is greater than zero, the feasible set is non-empty. Consequently, the problem in \eqref{ProblemMaxQ} always exists solutions. The proof is completed.
    \end{proof}

    % The feasibility of \eqref{Problemv1} is checked as in \cite{xxx}. 
    Note that this work focuses on addressing the infeasible problem of \eqref{Problemv1}. Therefore, the relationship between optimizing $R_{\textrm{tot}}(\boldsymbol{\Psi},\textbf{\textrm{p}})$ in \eqref{Problemv1-2} and optimizing $|\mathcal{Q}(\boldsymbol{\Psi},\textbf{\textrm{p}})|$ in \eqref{ProblemMaxQ} can be expressed as in Lemma \ref{ConflictMetrics} as follows.
    \begin{lemma}\label{ConflictMetrics}
    Given a fixed resource set, which includes the number of APs and the available power budget at each AP, the system is presumed to operate under infeasible conditions. The optimal solution for maximizing the total network throughput, $R_{\textrm{tot}}(\boldsymbol{\Psi},\textbf{\textrm{p}})$, without the rate constraint (similar to that considered in \eqref{Problemv1-2}) cannot ensure the optimality in the number of satisfied IoT devices, and vice versa. Consequently, optimizing the total network throughput, $R_{\textrm{tot}}(\boldsymbol{\Psi},\textbf{\textrm{p}})$, without the rate constraint (similar to that considered in \eqref{Problemv1-2}) and optimizing the number of satisfied IoT devices, $|\mathcal{Q},(\boldsymbol{\Psi},\textbf{\textrm{p}})|$ (similar to that considered in \eqref{ProblemMaxQ}) are conflict.
    \end{lemma}
    \begin{proof}
    We prove Lemma \ref{ConflictMetrics} by contradiction. Let denote $\{\boldsymbol{\Psi}^*, \textbf{\textrm{p}}^*\}$ be the optimal solution obtained by solving problem~\eqref{Problemv1-2}. The satisfied IoT devices set is then determined as in \eqref{QsetDeterminantEq}, and we have
    $$|\mathcal{Q}(\boldsymbol{\Psi}^*,\textbf{\textrm{p}}^*)|< |\mathcal{K}|.$$ 
    The above inequality comes from the fact that the system is infeasible circumstance.
    To provide service for more IoT devices, the system must have a proper policy to move IoT devices from the unsatisfied IoT device set $\mathcal{K} \setminus \mathcal{Q}(\boldsymbol{\Psi}^*,\textbf{\textrm{p}}^*)$ to the satisfied IoT device set $\mathcal{Q}(\boldsymbol{\Psi}^*,\textbf{\textrm{p}}^*)$, for example, providing a new resource allocation solution. By solving the problem in \eqref{ProblemMaxQ}, the proper policy generates a new solution $\{\boldsymbol{\Psi}^{'*}, \textbf{\textrm{p}}^{'*} \}$ to maximize the number of satisfied IoT devices, and we have
    $$R_{\textrm{tot}}(\boldsymbol{\Psi}^{'*},\textbf{\textrm{p}}^{'*})\leq R_{\textrm{tot}}(\boldsymbol{\Psi}^{*},\textbf{\textrm{p}}^{*}).$$ 
    The opposite approach can be implemented in a similar manner. Therefore, the optimal solution, $\{\boldsymbol{\Psi}^{'*},\textbf{\textrm{p}}^{'*}\}$, that maximizes the number of satisfied IoT devices cannot ensure the optimality of the total network throughput. Consequently, maximizing the number of satisfied IoT devices and maximizing the total network throughput are two conflict objectives.
    % Intuitively, sharing resources among IoT devices improves the received signal strength for some IoT devices suffering from higher interference than others that can be potentially upgraded. The achievement is acquired by the fact that the total network throughput is maximized by focusing on IoT devices with strong channel gains, degrading the individual rate of the remaining with weak channel conditions. Thus, these IoT devices may not be served under their requirements. Consequently, maximizing the number of satisfied IoT devices and maximizing the total data throughput are two conflict objectives.
    \end{proof}
    Based on Lemma \ref{ConflictMetrics}, the optimal number of satisfied IoT devices cannot be obtained by solving the problem \eqref{Problemv1-2}. Consequently, in order to serve more IoT devices, the system must sacrifice the total network throughput in exchange. However, if only the problem in \eqref{ProblemMaxQ} is optimized, the radio resource may not be exploited optimally. We, therefore, also consider the following dual-objective optimization problem
    \begin{equation} 
    \label{DualProblem_v1}
        \begin{aligned}
            && \underset{\substack{ \boldsymbol{\Psi}, \textbf{\textrm{p}} }}{\textrm{maximize} } & \quad  \left(|\mathcal{Q}(\boldsymbol{\Psi},\textbf{\textrm{p}})|,R_{\textrm{tot}}(\boldsymbol{\Psi},\textbf{\textrm{p}})\right),\\
            &&\textrm{subject to}\hspace{0.15cm}& \quad \underset{n\in \mathcal{N}}{\sum} \mu_{k,n}P_{n}\leq  P_{k}^{\textrm{max}},  \forall k\in \mathcal{K}, \\
            &&& \quad R_n(\boldsymbol{\Psi},\mathbf{p})\geq R_n^{\textrm{thr}}, \forall n\in \mathcal{Q}(\boldsymbol{\Psi},\textbf{\textrm{p}}),\\
            &&& \quad\sum_{k\in \mathcal{K}}\mu_{k,n}=1, \forall n\in \mathcal{N}.
        \end{aligned}
    \end{equation} 
    In this case, we only need to find a solution that simultaneously optimizes both the number of satisfied IoT devices and the total data throughput, given the higher priority for the number of satisfied IoT devices. The problem \eqref{DualProblem_v1} is non-convex due to the non-convexity of the objective function. Besides, the AP-IoT devices association is a discrete variable, and thus the problem in \eqref{DualProblem_v1} is considered as a mixed-integer non-linear program (MINLP). Unfortunately, this class of problems is nondeterministic polynomial-time (NP) complete. The problem requires time that is superpolynomial in the network size. Besides, in multi-objective optimizations, a feasible solution that simultaneously optimizes all the objectives does not typically exist. Therefore, achieving optimal solutions for two considered variables is inconsequential.
    
    To increase the number of satisfied IoT devices, we can provide more transmit power to un-satisfied IoT devices. However, the power budget of each AP is limited; thus, allocating more power to a specific IoT device will inevitably reduce the power available for the remaining IoT devices. If some IoT devices are served with insufficient data throughput, we can share power from satisfied IoT devices to un-satisfied IoT devices. However, this approach raises the challenge of determining which IoT devices should receive power and how much each IoT device should be given. Furthermore, increasing power for some IoT devices leads to higher interference with other IoT devices and reduces their data throughput. Subsequently, the improvement is hard to guarantee, while a reduction in the number of satisfied IoT devices could happen. On the other hand, iteratively adjusting the power of each IoT device can be impractical due to the significant computational cost involved. Nevertheless, it is worth noting that the satisfied IoT devices set is determined by comparing the achievable rate of each IoT device to the threshold, and thus, the set $\mathcal{Q}(\boldsymbol{\Psi},\textbf{\textrm{p}})$ is directly related to $R_{\textrm{tot}}(\boldsymbol{\Psi},\textbf{\textrm{p}})$. Therefore, intuitively, by increasing the total data throughput of all IoT devices controllably, we can expect that some IoT devices can meet their requested services, thus increasing the number of satisfied IoT devices. In the next section, we will propose a solution to increase the number of satisfied IoT devices by optimizing the total data throughput.
    \begin{remark}
        The APs in the considered IoT network as mentioned in Remark \ref{remark1} are limited in computational ability and other resources. Cooperation among APs to boost network performance is challenging due to the limited resources and backhaul signaling. A server (central processing unit) is usually utilized to gather information about data transmission, APs, and IoT devices' status to manage the network operation. Regarding the total data throughput maximization problem expressed in \eqref{Problemv1}, the infeasible problem occurs when the system cannot find a solution to all IoT devices with their service requirements. Consequently, the whole system is corrupted. In contrast, problem \eqref{ProblemMaxQ} and \eqref{DualProblem_v1} ensure that the system always works and serves at least a subset of IoT devices, and \eqref{DualProblem_v1} still maximize the total data throughput.
        
        % As our concentration is on finding a methodology for the system under infeasible circumstances, the perfect CSI is assumed to be available at the both APs and IoT devices, similar to many other related works in the literature such as \cite{Gbadamosi2022,Pang2008Empty_nonEmpty_game, Obiedollah2020_FeasibilityCheck,Stridh2006Interference_CongestionControl,Antonioli2019IBC, Yu2012CoordinatedScheduling, Chien2020CongestionICC, Chien2021ScatteringChannelICC, van2021uplink,Chi2005}. The analysis on imperfect CSI conditions is left for future work.
    \end{remark}
    \begin{remark}
        Solving problems in \eqref{ProblemMaxQ} and \eqref{DualProblem_v1} requires the selection of IoT devices from the set of IoT devices joining the network. Based on the specific requirement of each network, IoT devices can be selected based on the priority of the system or users, \cite{Antonioli2019IBC}. In this work, we assume no priority between IoT devices, and thus, the system will attempt to optimize the number of satisfied devices without considering any specific individual IoT device.

        Additionally, while the assumption of perfect CSI is idealistic, it is crucial for the initial development and validation of our proposed framework. This assumption enables a focused analysis of the optimization challenges without the added complexity of imperfect information, laying the foundation for future studies to address more realistic conditions.
    \end{remark}
     
    \section{Proposed Algorithm}\label{OptimiaztionApproach}
    It is worth noting that the primal problem in \eqref{Problemv1} is infeasible, and thus we propose to solve the number of satisfied IoT devices maximization problem in \eqref{ProblemMaxQ} and the dual-objectives function problem in \eqref{DualProblem_v1}, which are inequivalent to the primal problem in \eqref{Problemv1}. Due to the adversity of the system under infeasible circumstances and the difficulties of the proposed problems \eqref{ProblemMaxQ} and \eqref{DualProblem_v1} as discussed from previous section, finding the optimal solution for the system cannot be ensured. Consequently, by prioritizing serving as many IoT devices as possible, this section proposes solutions to find a good solution to the system under infeasible circumstances.
    
    \subsection{The number of satisfied IoT devices optimization}
    Due to the combinatorial nature of the AP-IoT devices association with a total of $N^K$ candidates, the problem in \eqref{ProblemMaxQ} is extremely difficult to obtain the optimal solution. To overcome this issue, we propose a modified BB algorithm based on the branch-and-bound method. In detail, we construct a BB tree that starts from an arbitrary IoT device, which is considered the root of the BB tree. The system successively testifies $K$ nodes corresponding to $K$ APs. If the power consumption needed to meet an IoT device's required data throughput falls below the power budget of the considered AP, the current node will be added to the first level of the BB tree. 

    The BB tree then goes to the next level and considers a new IoT device. At the $n$-th level, the IoT device of interest will be attached consecutively to every node of the $n-1$-th levels. The IoT device at the $n$-th level is then traversed to all $K$ nodes, and thus we can obtain the temporary AP-IoT devices association of $n$ IoT devices, $\boldsymbol{\Psi}^{\text{temp}}_n$. The minimum power consumption of these $n$ IoT devices is achieved by solving the following problem
    \begin{equation}\label{Infeasible checker 1}
        \begin{aligned}
            && \underset{\substack{ \mathbf{p}_n }}{\textrm{minimize}} & \quad \sum\limits_{i=1}^{n} P_i, \\
            &&\textrm{subject to}\hspace{0.15cm}&\quad R_i(\boldsymbol{\Psi}^{\text{temp}}_n,\mathbf{p}_n)\geq R_i^{\textrm{thr}}, \forall i\in \{1,...,n\},
            % &\underset{n\in \mathcal{N}}{\sum} \mu_{k,n}P_{n}\leq  P_{k}^{\textrm{max}},  \forall k\in \mathcal{K}, \\
            % &\frac{ |\boldsymbol{\mu}_{n}^T \boldsymbol{h}_{n}|^2 P_{n}}{\underset{{n'\neq n,n'\in \mathcal{N}}}{\sum}|\boldsymbol{\mu}_{n'}\boldsymbol{h}_{n}|^2 P_{n'}+\sigma^2}\geq 2^{R_n^{\textrm{thr}}}-1, \forall n\in \mathcal{N},\label{V1c} \\
            % &\sum_{k\in \mathcal{K}}\mu_{k,n}=1, \forall n\in \mathcal{N},
        \end{aligned}
    \end{equation} 
    The problem in \eqref{Infeasible checker 1} is a linear program, we can quickly obtain the optimal solution, $\mathbf{p}^*_n$. Note that the system then testifies the power budgets of all APs
    \begin{equation}\label{BB Power constraint}
        \sum\limits_{i=1}^n P_i\mu_{k,i} \leq P_k^{\max}, \forall k\in \mathcal{K}
    \end{equation}
    If the power constraints in \eqref{BB Power constraint} are violated, the considered node and its children will be pruned from the BB tree. The process is repeated until reaching the last level of the BB tree or all new nodes are pruned. The detailed procedure is summarized as in Algorithm~\ref{BB Algorithm}
    \begin{algorithm}[t]
		\caption{Modified Branch-and-Bound for the number of satisfied IoT devices optimization (the modified BB algorithm).} \label{BB Algorithm} 
		\begin{algorithmic}[1] 
		      \STATE \textbf{Input}: The system parameters $\mathcal{N}$, $\mathcal{K}$, $B$, $\sigma$, $P^{\max}$, the channel state information $\textbf{h}_n, \forall n\in \mathcal{N}$.
			\STATE Generate an arbitrary order of $N$ IoT devices. Set $n=0$.
			\REPEAT 
			\STATE {$n=n+1$.}
			\REPEAT
                \FOR {k=1:K} 
			\STATE {Associate the $n$-th IoT device to the $k$-th AP.}
                \STATE {Form the temporary AP-IoT devices association of $n$ IoT devices, $\boldsymbol{\Psi}^{\text{temp}}_n$.}
			\STATE {Solve the problem in \eqref{Infeasible checker 1} to achieve     $\mathbf{p}^*_n$.}
                \IF {Constraints in \eqref{BB Power constraint} are not satisfied.}
                \STATE {Pruned node.}
                \ENDIF
                \ENDFOR
			\UNTIL {All nodes of the $(n-1)$-th level are testified.}
			\UNTIL {$n=N$ or all new nodes are pruned.} 
			\STATE \textbf{Output}: $n^*$ number of satisfied IoT devices, the AP-IoT devices association, $\boldsymbol{\Psi}_{n^*}$, and the power allocation $\textrm{\textbf{p}}^*_{n^*}$. 
		\end{algorithmic} 
    \end{algorithm} 

    The performance of the modified BB algorithm \ref{BB Algorithm} will be affected by the order of IoT devices on the BB tree. Therefore, the modified BB algorithm that starts with an arbitrary order of IoT devices cannot ensure the optimal solution. Searching for the optimal order of IoT devices for the BB tree requires the complexity of $\mathcal{O}(N^K)$. As a result, employing an arbitrary order of IoT devices in the BB tree is deemed acceptable for identifying a favorable solution for the system, even under infeasible conditions. However, the BB tree must reach the last level with the existence of at least one node to ensure the non-empty feasible domain for the problem in \eqref{Problemv1}. Subsequently, the modified BB algorithm can be used to check the infeasibility of the original problem in \eqref{Problemv1}.
    \subsection{Proposed solution to the dual-objective optimization problem in \eqref{DualProblem_v1}}\label{ProposedSolution}
    We recall that if the system optimizes the sum of data throughput without meeting all the IoT devices' QoS requirements, the power budget would be allocated to the IoT devices with good channel conditions. In this case, the power allocated to IoT devices with poor channel conditions will be reduced, which will potentially cause low data throughput for such IoT devices. Thus, a solution to improve the fairness among IoT devices while maximizing the overall network throughput is needed. The high-level idea here is that we will allocate more power to IoT devices with poor channel conditions to ensure fairness while ensuring high network throughput for the whole system. In order to share more power to unsatisfied IoT devices, we suggest setting the data throughput requirement of satisfied IoT devices equal to their requested levels. In this way, the system only allocates a certain amount of the power budget to these IoT devices to satisfy the request. The remaining power budget is allocated to the unsatisfied IoT devices in order to improve the total network throughput. When the system reallocates the resources with fixed data throughput constraints, some IoT devices might still receive more data throughput than needed. In this case, the resources from these IoT devices can be shared with other IoT devices again. Therefore, to obtain the final solution, we repeat two steps: $(1)$ updating the satisfied IoT devices by reallocating the power budget; $(2)$ solving the data throughput maximization problem with fixed data throughput constraints.
    
    We first determine an initial satisfied IoT device set, denoted by $\mathcal{Q}(\boldsymbol{\Psi}^{*,(0)},\textbf{\textrm{p}}^{*,(0)})$. Therein, the initial AP association, $\boldsymbol{\Psi}^{*,(0)}$, is determined based on the best large-scale fading selection. Particularly, $\mu^{*,(0)}_{k,n}=1$ if $\zeta_{k,n}\leq \zeta_{j,n}, \forall j\in \mathcal{K}$, and $\mu^{*,(0)}_{j,n}=0, \forall j\neq k, j\in \mathcal{K}$. Here, $\zeta_{k,n}$ is the path loss between the $k$-th AP and the $n$-th IoT device. Besides, the initial power allocation, $\textbf{\textrm{p}}^{*,(0)}$, is obtained by solving the following problem
    \begin{subequations} \label{MaxRwithoutRconstraint_v2}
        \begin{align}
            && \underset{\substack{\textbf{\textrm{p}} }}{\textrm{maximize} } & \quad  R_{\textrm{tot}}(\textbf{\textrm{p}}),\\
            &&\textrm{subject to}\hspace{0.15cm}& \quad \underset{n\in \mathcal{N}}{\sum} \mu^{*,(0)}_{k,n}P_{n}\leq  P_{k}^{\textrm{max}}, \forall k\in \mathcal{K}.
        \end{align}
    \end{subequations} 
    We note that the non-convexity of the objective function is challenging to solve problem~\eqref{MaxRwithoutRconstraint_v2}. Our approach is to transform \eqref{MaxRwithoutRconstraint_v2} into a convex optimization problem by applying the log approximation as introduced below:
    \begin{equation}\label{LogRelax}
		\log_2 (1+z)\geq \alpha(z') \log_2 (z) + \beta(z'),
    \end{equation}
    where $z$ and $z'$ are two non-negative values. In \eqref{LogRelax}, the log transformation factors $\alpha(z')$ and $\beta(z')$ are defined as follows:
    \begin{align}
         &\alpha(z') =\frac{z'}{1+z'},\\
         & \beta(z')=\log_2(1+z')-\frac{z'}
         {1+z'}\log_2(z').
    \end{align}
    We emphasize that the bound in \eqref{LogRelax} is tight, i.e., the inequality holds with the equality as $z=z'$. We now apply the approximation in \eqref{LogRelax} to provide the lower bound of the data throughput $R_{n}$ as follows:
    \begin{equation} \label{LowerBoundRate1}
        R_{n}(\textbf{\textrm{p}})\geq \tilde{R}_{n}(\textbf{\textrm{p}})= B\left(\alpha_{n}(\textbf{\textrm{p}}') \log_2 (\gamma_{n}(\textbf{\textrm{p}}))+\beta_{n}(\textbf{\textrm{p}}')\right),
    \end{equation}
    where the following definitions hold
    \begin{equation}
        \begin{aligned}\label{LogApproximation1}
            &\alpha_{n}(\textbf{\textrm{p}}')=\frac{\gamma_{n}(\textbf{\textrm{p}}')}{1+\gamma_{n}(\textbf{\textrm{p}}')},\\
            &\beta_{n}(\textbf{\textrm{p}}')=\log_2(1+\gamma_{n}(\textbf{\textrm{p}}'))-\frac{\gamma_{n}(\textbf{\textrm{p}}')}{1+\gamma_{n}(\textbf{\textrm{p}}')}\log_2(\gamma_{n}(\textbf{\textrm{p}}')).
        \end{aligned} 
    \end{equation}
    with recalling that $R_{n}(\textbf{\textrm{p}})= \tilde{R}_{n}(\textbf{\textrm{p}})$ as $\textbf{\textrm{p}}=\textbf{\textrm{p}}'$. Let us define 
    \begin{align}
        &\boldsymbol{\alpha}(\textbf{\textrm{p}}')=\{\alpha_{n}(\textbf{\textrm{p}}')\}\in \mathcal{R}^{N\times 1},\\
        & \boldsymbol{\beta}(\textbf{\textrm{p}}')=\{\beta_{n}(\textbf{\textrm{p}}')\}\in \mathcal{R}^{N\times 1},
    \end{align}
    then the problem in~\eqref{MaxRwithoutRconstraint_v2} is approximated as follows:
    \begin{equation} \label{MaxRwithoutRconstraint_v3}
        \begin{aligned}
            && \underset{\substack{ \textbf{\textrm{p}} }}{\textrm{maximize} } & \quad  \tilde{R}_{\textrm{tot}} (\textbf{\textrm{p}})=\sum_{n=1}^N\tilde{R}_{n}(\textbf{\textrm{p}}),\\
            &&\textrm{subject to}\hspace{0.15cm}&\quad 0\leq \underset{n\in \mathcal{N}}{\sum} \mu^{*,(0)}_{k,n}P_{n}\leq  P_{k}^{\textrm{max}}, \forall k\in \mathcal{K}.
        \end{aligned}
    \end{equation} 
    An iterative algorithm is proposed to solve problem~\eqref{MaxRwithoutRconstraint_v2} for a tight approximation. Particularly, there are two main steps comprising: $(1)$ optimizing $\tilde{R}_{\textrm{tot}}(\textbf{\textrm{p}})$ while considering the log transformation factors as constant values; and $(2)$ updating the log transformation factors according to the new power vector. Therefore, in the $i$-th iteration, we solve the following problem
    \begin{equation} \label{MaxRwithoutEqualRconstraint_v4}
        \begin{aligned}
            && \underset{\substack{ {\textbf{\textrm{p}}}}}{\textrm{maximize}} & \quad  \sum_{n=1}^NB\big(\alpha_{n}({\textbf{\textrm{p}}}[i-1])\log_2 (\gamma_{n}({\textbf{\textrm{p}}}))\\
            &&& \hspace{1.5cm}+\beta_{n}({\textbf{\textrm{p}}}[i-1])\big),\\
            &&\textrm{subject to}\hspace{0.15cm}&\quad 0\leq \underset{n\in \mathcal{N}}{\sum} \mu^{*,(0)}_{k,n}{P}_{n}\leq  P_{k}^{\textrm{max}}, \forall k\in \mathcal{K}.
        \end{aligned} 
    \end{equation} 
    We observe that problem~\eqref{MaxRwithoutEqualRconstraint_v4} has a hidden-convex form. For such, we apply the transformation $P_n=\exp(\bar{P}_n)$ to transform \eqref{MaxRwithoutEqualRconstraint_v4} into a convex problem and obtain the global solution by using a general-purpose optimization toolbox. The log approximation method is terminated when the variation between two consecutive iterations is smaller than a sufficient small tolerance $\epsilon_1$. The procedure of solving \eqref{MaxRwithoutRconstraint_v2} is provided in Algorithm \ref{LogApproximation1}, and its convergence is shown in Lemma~\ref{LogLemma_NoRateConstraint}.
    \begin{algorithm}[t]
		\caption{Proposed solution to problem~\eqref{MaxRwithoutRconstraint_v2}} \label{LogApproximation1} 
		\begin{algorithmic}[1] 
		  \STATE \textbf{Input}: System parameters $\mathcal{N}$, $\mathcal{K}$, $B$, $\sigma$, $P^{\max}$; Channel state information $\textbf{h}_n, \forall n\in \mathcal{N}$, the AP association $\boldsymbol{\Psi}^{*,(0)}$, the initial power vector $\textbf{\textrm{p}}[0]$, and the tolerance $\epsilon_1$.
            \STATE Set $i=0$ and calculate $\boldsymbol{\alpha}(\textbf{\textrm{p}}[0])$ and $\boldsymbol{\beta}(\textbf{\textrm{p}}[0])$.
            \REPEAT 
            \STATE {$i\gets i+1$.} 
            \STATE {Solve the problem \eqref{MaxRwithoutEqualRconstraint_v4} and obtain $\textbf{\textrm{p}}[i]$.} 
            \STATE {Update $\boldsymbol{\alpha}(\textbf{\textrm{p}}[i])$ and $\boldsymbol{\beta}(\textbf{\textrm{p}}[i])$.} 
            % \STATE Calculate $\nabla=\|\textbf{\textrm{p}}[i]-\textbf{\textrm{p}}[i-1]\|\geq \epsilon_1$
            \UNTIL {$\|\textbf{\textrm{p}}[i]-\textbf{\textrm{p}}[i-1]\|\leq \epsilon_1$.} 
            \STATE \textbf{Output}: $\textrm{\textbf{p}}^{*,(0)}$.
		\end{algorithmic}
    \end{algorithm} 
    \begin{lemma}\label{LogLemma_NoRateConstraint}
    \textcolor{black}{By solving problem~\eqref{MaxRwithoutEqualRconstraint_v4} and updating the log transformation factors iteratively as in Algorithm \ref{LogApproximation1}, the objective function of \eqref{MaxRwithoutRconstraint_v2} will be non-decreased and converged to a fixed point.}
        % The final solution satisfies the KKT conditions of \eqref{MaxRwithoutRconstraint_v2}.
    \end{lemma}
    \begin{proof}
    \textcolor{black}{The proof is available in Appendix~\ref{ProofLemmaLogApproximation}.}
    \end{proof}
    \textcolor{black}{Lemma~\ref{LogLemma_NoRateConstraint} confirms the convergence of     Algorithm~\ref{LogLemma_NoRateConstraint} thanks to the exploitation of the log approximation and convex feasible set. For given $\boldsymbol{\Psi}^{*,(0)}$ and $\textbf{\textrm{p}}^{*,(0)}$, the satisfied IoT devices set is defined as follows:
   \begin{equation}
   \mathcal{Q}^{*}(\boldsymbol{\Psi}^{*,(0)}, \textbf{\textrm{p}}^{*,(0)})=\{n|R_n(\boldsymbol{\Psi}^{*,(0)}, \textbf{\textrm{p}}^{*,(0)})\geq R_n^{\textrm{thr}}, \forall n\in \mathcal{N}\},
   \end{equation}
    which indicates the initial set of satisfied IoT devices. We stress that $\mathcal{Q}^{*}(\boldsymbol{\Psi}^{*,(0)}, \textbf{\textrm{p}}^{*,(0)}) \subseteq \mathcal{K}$ and its features are observed as follows.}
   \begin{remark}
    After solving problem~\eqref{MaxRwithoutRconstraint_v2} by exploiting Algorithm~\ref{LogApproximation1} to obtain the power coefficients $\textbf{\textrm{p}}^{*,(0)}$ and together with $\boldsymbol{\Psi}^{*,(0)}$, there are two special observations of the satisfied IoT device set. If $\mathcal{Q}^{*}(\boldsymbol{\Psi}^{*,(0)}, \textbf{\textrm{p}}^{*,(0)})= \emptyset$, it implies that no IoT device is satisfied its data throughput requirements at the initial solution. To improve the service management, we allocate the maximum transmit power to an IoT device with the smallest path loss (strongest channel gain), $P_n=P_k^{\max}$ if $\zeta_{n,k}\leq \zeta_{i,j}, \forall i\in \mathcal{N}, j\in \mathcal{K}$. In contrast, we allocate the zero transmit power to the remaining IoT devices, $P_{n'}=0, \forall n'\neq n$. By this policy, the prioritized IoT device will attain the highest received signal strength and no mutual interference. Once the IoT device meets the minimum data throughput, we add the IoT device to the initial satisfied IoT device set. Otherwise, the selected IoT device cannot meet the data throughput, and the system cannot serve IoT devices with their service requirements. If $|\mathcal{Q}^{*}(\boldsymbol{\Psi}^{*,(0)}, \textbf{\textrm{p}}^{*,(0)})|=N$, it means that all IoT devices are served with requested services.
    \end{remark}
    We then solve the total data throughput maximization problem with the fixed quality of service constraints from satisfied IoT devices with respect to the limited power constraint in \eqref{V1b} and the AP-IoT devices association in \eqref{V1d} as below
	\begin{subequations} \label{MaxRwithoutRconstraint_v1}
		\begin{align}
			&& \underset{\substack{ \boldsymbol{\Psi}, \textbf{\textrm{p}} }}{\textrm{maximize}} & \quad  R_{\textrm{tot}}(\boldsymbol{\Psi}, \textbf{\textrm{p}}),\\
			&&\textrm{subject to}\hspace{0.15cm}& \quad\eqref{V1b},\eqref{V1d}, \\
			&&&\quad R_n(\boldsymbol{\Psi},\textbf{\textrm{p}})= \xi_n^{\textrm{thr}}, \forall n\in \mathcal{Q}^*(\boldsymbol{\Psi}^{*,(t-1)}, \textbf{\textrm{p}}^{*,(t-1)}) \label{MaxRwithoutRconstraint_v1:EqualRConstraint} 
		\end{align}
    \end{subequations}
    where $t$ is the iteration index. Due to highly coupling of mutual interference, we introduce the constraint \eqref{MaxRwithoutRconstraint_v1:EqualRConstraint}, $\xi_n^{\textrm{thr}}=R_n^{\textrm{thr}}+R_n^{\textrm{thr}}\tau$, to avoid the undesired situations that the data throughput values of IoT devices in $\mathcal{Q}^{*}$ are slightly smaller than the threshold, in which $\tau$ is sufficiently small tolerable constant. The constraint \eqref{MaxRwithoutRconstraint_v1:EqualRConstraint} is to ensure that the system can serve the satisfied IoT devices in $\mathcal{Q}^{*}$ with only their service requirements. With a finite power level $\textrm{P}^{\max}$, the remaining power of each AP should be allocated to other IoT devices, and it is expected that their data throughput will increase and meet the requirements. The set of satisfied IoT devices will be updated after solving the problem in \eqref{MaxRwithoutRconstraint_v1}, 
    \begin{equation}
        \mathcal{Q}^*(\boldsymbol{\Psi}^{*,(t)}, \textbf{\textrm{p}}^{*,(t)})=\{n|R_n(\boldsymbol{\Psi}^{*,(t)}, \textbf{\textrm{p}}^{*,(t)})\geq R_n^{\textrm{thr}}, \forall n\in \mathcal{N}\}.
    \end{equation}
    As mentioned before, the problem in \eqref{MaxRwithoutRconstraint_v1} is difficult to simultaneously optimally obtain both $\boldsymbol{\Psi}$ and $\textbf{\textrm{p}}$. Thus, we suggest sequentially obtaining each variable while considering the others fixed. In the following subsections, we will provide the solutions to the power allocation and the AP association problems. The solutions to sub-problems guarantee non-decreasing the number of satisfied IoT devices and the total data throughput.
    \begin{remark}
        The modified BB algorithm can be used to initialize the set of satisfied IoT devices. However, this initialization will force the system to provide services to the set of IoT devices that have already been allocated resources to achieve the required data throughput only. Therefore, fixing the data throughput of these IoT devices and optimizing the resource allocation again will be trivial.
    \end{remark}
    \subsection{Total Data Throughput Maximization Problem With The Quality of Service Constraints}\label{MaxQset}
    In this subsection, we find the solution to problem~\eqref{MaxRwithoutRconstraint_v1} in an iterative manner. One observes that the AP association for the IoT devices impacts mutual interference and the power allocation under the limited power budget. Each IoT device can select one of the $K$ APs, a total of $N^K$ possibilities should be exhaustively searched to obtain the global optimum. Consequently, obtaining both the AP association and the power allocation is nontrivial for a large-scale network. Therefore, we propose transforming the problem into two sub-problems corresponding to solving the power allocation and the AP-IoT device association. Optimizing the two sub-problems is not equivalent to optimizing the original problem. To achieve a good solution, an alternative algorithm is needed to acquire the solution to each optimization variable iteratively. For convenience, we remove the iteration index at all notations in this sub-section.
    % \vspace{-0.25cm}
    % Our approach to this problem is to iteratively obtain each variable while considering the remaining one fixed. Noting that the solutions to the AP association problem and the power allocation must ensure the improvement in the total data throughput of all IoT devices. 
    % Therein, the system will find the AP association of all IoT devices that can improve the total data throughput of the system while still guaranteeing the limited power budget at each AP. 
    \subsubsection{Power allocation with equal data throughput constraints}\label{Subsection:PowerAllocation}
    For a given AP association, the power allocation optimization is rewritten below
    \begin{subequations} \label{MaxRwithoutRconstraint_P_v1}
        \begin{align}
            && \underset{\substack{\textbf{\textrm{p}} }}{\textrm{maximize} } & \quad  R_{\textrm{tot}}(\textbf{\textrm{p}}),\\
            &&\textrm{subject to}\hspace{0.15cm}& \quad \underset{n\in \mathcal{N}}{\sum} \mu_{k,n}P_{n}\leq  P_{k}^{\textrm{max}},  \forall k\in \mathcal{K},\label{MaxRwithoutRconstraint_P_v1: V1b}, \\
            &&&\quad R_n(\boldsymbol{\Psi},\textbf{\textrm{p}})= \xi_n^{\textrm{thr}}, \forall n\in \mathcal{Q}^* \label{MaxRwithoutRconstraint_v1:EqualRConstraint: V1c}.
            % &&&\quad R_n(\boldsymbol{\Psi},\textbf{\textrm{p}})= \xi_n^{\textrm{thr}}, \forall n\in \mathcal{Q}^{*} \label{MaxRwithoutRconstraint_v1:EqualRConstraint} 
        \end{align}
    \end{subequations}
    We use the log approximation method to get the lower bound on the data throughput formula and transform problem~\eqref{MaxRwithoutRconstraint_P_v1} into a tractable form, which is given by
    \begin{subequations} \label{MaxRwithoutRconstraint_v3}
        \begin{align}
            && \underset{\substack{ \textbf{\textrm{p}} }}{\textrm{maximize} } & \quad  \tilde{R}_{\textrm{tot}}(\textbf{\textrm{p}})=\sum_{n=1}^N\tilde{R}_{n}(\textbf{\textrm{p}}),\\
            &&\textrm{subject to}\hspace{0.15cm}&\quad 0\leq \underset{n\in \mathcal{N}}{\sum} \mu_{k,n}P_{n}\leq  P_{k}^{\textrm{max}}, \forall k\in \mathcal{K},\\
            &&&\quad \alpha_{n}(\bar{\textbf{\textrm{p}}})\log_2 (\gamma_{n}(\bar{\textbf{\textrm{p}}}))+\beta_{n}(\bar{\textbf{\textrm{p}}})= \xi_n^{\textrm{thr}}, \forall n\in \mathcal{Q}^{*}.
            % \label{MaxRwithEqualRconstraint_v5:EqualRConstraint}
        \end{align}
    \end{subequations}  
    The iterative algorithm as in Algorithm \ref{LogApproximation1} is also applied to optimize the objective of the problem in \eqref{MaxRwithoutRconstraint_P_v1}. 
    At every iteration, we need to solve the following problem
    % \begin{figure*}[t]
        \begin{subequations} \label{MaxRwithEqualRconstraint_v5}
            \begin{align}
                &\underset{\substack{ \bar{\textbf{\textrm{p}}}}}{\textrm{maximize}}  \quad  \sum_{n\in \mathcal{N}/\mathcal{Q}^{*}}^NB\Big(\alpha_{n}(\bar{\textbf{\textrm{p}}}[i-1])\log_2 (\gamma_{n}(\bar{\textbf{\textrm{p}}}))\nonumber \\
                &\hspace{3.4cm}+\beta_{n}(\bar{\textbf{\textrm{p}}}[i-1])\Big),\\
                &\textrm{subject to} \nonumber\\
                &\quad 0\leq \underset{n\in \mathcal{N}}{\sum} \mu_{k,n}\exp(\bar{P}_{n})\leq  P_{k}^{\textrm{max}}, \forall k\in \mathcal{K},\\
                &\quad \alpha_{n}(\bar{\textbf{\textrm{p}}}[i-1])\log_2 (\gamma_{n}(\bar{\textbf{\textrm{p}}}))+\beta_{n}(\bar{\textbf{\textrm{p}}}[i-1])= \xi_n^{\textrm{thr}},\forall n\in \mathcal{Q}^{\ast}, \label{MaxRwithEqualRconstraint_v5:EqualRConstraint}
            \end{align}
        \end{subequations}    
    % \end{figure*}
    where $P_n=\exp(\bar{P}_n)$, and the data throughput of each satisfied IoT device in $\mathcal{Q}^{\ast}$ is removed from the objective function of problem~\eqref{MaxRwithEqualRconstraint_v5} since the network treats their services as constants. However, due to the constraints \eqref{MaxRwithEqualRconstraint_v5:EqualRConstraint}, problem~\eqref{MaxRwithEqualRconstraint_v5} is non-convex. At this step, we suggest applying the partial Lagrangian function defined as follows:
    \begin{equation}\label{(partial)LagrangianFunction}
        \begin{aligned}
            L(\bar{\textbf{\textrm{p}}},\boldsymbol{\theta})=&\sum_{n\in \mathcal{N} 
     }B\left(\alpha_{n}(\bar{\textbf{\textrm{p}}}[i-1])\log (\gamma_{n}(\boldsymbol{\Psi},\bar{\textbf{\textrm{p}}}))+\beta_{n}(\bar{\textbf{\textrm{p}}}[i-1])\right)\\
            &-\sum_{k=1}^{K}\theta_{k}(\sum_{n=1}^{N}\mu_{k,n}\exp(\bar{\text{P}}_n)-P_k^{\textrm{max}}),
        \end{aligned}
    \end{equation}
    where $\theta_k$ is the dual variable corresponding to the
    power constraint of the $k$-th AP, and $\boldsymbol{\theta}=[\theta_1,...,\theta_N]$. The equal data throughput constraints, \eqref{MaxRwithEqualRconstraint_v5:EqualRConstraint}, are not introduced to the multipliers and are used as back-tracking conditions to update the power. Based on the Lagrangian function in \eqref{(partial)LagrangianFunction}, the dual optimization problem is formulated as:
    \begin{equation}\label{DualProblem}  
		\begin{aligned}
			&& \underset{\substack{\boldsymbol{\theta} }}{\mathrm{minimize}} & \quad g(\boldsymbol{\theta}) = \underset{\bar{\textbf{\textrm{p}}}}{\textrm{maximize}} \hspace{0.1cm}L(\bar{\textbf{\textrm{p}}},\boldsymbol{\theta}),\\
			&&\mbox{subject to}&\quad \theta_n \geq 0, \forall n\in \mathcal{N}.\\
            &&& \quad \eqref{MaxRwithEqualRconstraint_v5:EqualRConstraint}.
		\end{aligned} 
    \end{equation} 
    We emphasize that the Lagrangian function in \eqref{(partial)LagrangianFunction} sets an upper bound to the objective function of problem~\eqref{MaxRwithoutRconstraint_v3}. We, therefore, optimize $L(\bar{\textbf{\textrm{p}}},\boldsymbol{\lambda})$ with respect to $\bar{P}_n, n\in \mathcal{N}/\mathcal{Q}^{*}$. For such, the corresponding derivative is computed as
    \begin{equation}\label{Lagrangian derivative} 
    	\begin{aligned}  
    		&\frac{\partial L(\bar{\textbf{\textrm{p}}},\boldsymbol{\theta})}{\partial \bar{P}_n}= \frac{1}{\ln(2)}B\alpha_{n}(\bar{\textbf{\textrm{p}}}[i-1])-\sum_{n'=1,\atop n' \neq n}^{N} B\alpha_{n'}(\bar{\textbf{\textrm{p}}}[i-1])\\
    		&\times \frac{\sum_{k=1}^K |\mu_{k,n} g_{k,n'}|^2\exp (\bar{P}_n)}{\ln(2)\left(\sum^N_{j=1, j\neq n'}\sum_{k=1}^K|\mu_{k,j}g_{k,n'}|^2\exp(\bar{P}_j)+\sigma^2\right)}\\ 
    		&-\sum_{k=1}^{K}\theta_{k}\mu_{k,n}\exp(\bar{\text{P}}_n)=0, \forall n\in \mathcal{N}/\mathcal{Q}^{*}.
    	\end{aligned}
    \end{equation}
    From \eqref{Lagrangian derivative}, after a back transformation $P_n=\exp(\bar{P}_n)$, a set of fixed-point solutions to the power allocation of unsatisfied IoT devices is formed as in \eqref{FixedPointEq1:UnsatisfiedIoT devices}. Meanwhile, the back-tracking condition  \eqref{MaxRwithEqualRconstraint_v5:EqualRConstraint} is rewritten as
    \begin{equation}
        \begin{aligned}
            \frac{P_n\sum\limits_{k=1}^K|\mu_{k,n}g_{k,n}|^2}{\underset{n'\neq n,n'\in \mathcal{N}}{\sum}\sum\limits_{k=1}^K|\mu_{k,n'}g_{k,n}|^2P_{n'}+\sigma^2}=&2^{\frac{\frac{\xi_n^{\textrm{thr}}}{B}-\beta_{n}({\textbf{\textrm{p}}}[i-1])}{\alpha_{n}({\textbf{\textrm{p}}}[i-1])}},\forall n\in \mathcal{Q}^{*}.
        \end{aligned}
    \end{equation}
   The power allocation of satisfied IoT devices in $\mathcal{Q}^{*}$ is then determined by a set of fixed-point solutions as in \eqref{FixedPointEq1:SatisfiedIoT devices}. 
	\begin{figure*}[t]
        \begin{subequations}\label{FixedPointEq2}
            \begin{align}
                P_n=&\frac{B\alpha_{n}({\textbf{\textrm{p}}}[i-1])}{\underset{n'\neq n}{\sum} B\alpha_{n'}({\textbf{\textrm{p}}}[i-1])\frac{\sum\limits_{k=1}^K \abs{\mu_{k,n} g_{k,n'}}^2}{\underset{j\neq n'}{\sum}\sum\limits_{k=1}^K\abs{\mu_{k,j}g_{k,n'}}^2P_j+\sigma^2}+\sum\limits_{k=1}^{K}\ln(2)\theta_{k}\mu_{k,n}}=I_n(\textbf{\textrm{p}}), \forall n\in \mathcal{N}/\mathcal{Q}^{*}. \label{FixedPointEq1:UnsatisfiedIoT devices}\\
            	P_n=&\frac{2^{\frac{\frac{\xi_n^{\textrm{thr}}}{B}-\beta_{n}({\textbf{\textrm{p}}}[i-1])}{\alpha_{n}({\textbf{\textrm{p}}}[i-1])}}\left(\underset{n'\neq n}{\sum}\sum\limits_{k=1}^K\abs{\mu_{k,n'}g_{k,n}}^2P_{n'}+\sigma^2\right)}{\sum\limits_{k=1}^K\abs{\mu_{k,n}g_{k,n}}^2}=I_n(\textbf{\textrm{p}}), \forall n\in \mathcal{Q}^{*} .\label{FixedPointEq1:SatisfiedIoT devices}
            \end{align}
        \end{subequations}
		\hrule
		\vspace*{0.1cm}
	\end{figure*}
     According to the IoT device status, i.e., either satisfied or unsatisfied, \eqref{FixedPointEq1:UnsatisfiedIoT devices} and \eqref{FixedPointEq1:SatisfiedIoT devices} can be written as 
    \begin{equation}\label{PointWiseVector}
        \textbf{\textrm{p}}=\textbf{I}(\textbf{\textrm{p}}),
    \end{equation}
    where the following definition holds
    \begin{equation}
    \textbf{I}(\textbf{\textrm{p}})=[I_1(\textbf{\textrm{p}}),...,I_N(\textbf{\textrm{p}})]^T,
    \end{equation}
    and $I_n(\textbf{\textrm{p}})$ is defined as in \eqref{FixedPointEq1:UnsatisfiedIoT devices} and \eqref{FixedPointEq1:SatisfiedIoT devices} according to the served data throughput. From initial transmit powers in the feasible domain, $\mathbf{p}^{(0)}$, we then have the following update  
	\begin{equation}\label{FixedPointWithEqRateConstraint}
	    \textbf{\textrm{p}}(s+1)=\textbf{I}(\textbf{\textrm{p}}(s)),
	\end{equation}
	where $s$ is the iteration index. The Lagrangian multipliers are updated to penalize the violation of limited power constraints as follows:
     \begin{equation} \label{eq:Lagrange}
     \theta_{k}{(s+1)}=\left[\theta_{k}{(s)}+\epsilon_\theta\left(\sum_{n=1}^N\mu_{k,n}\exp(\bar{P}_n)-P_k^{\max}\right)\right]. 
    \end{equation}
    where $\epsilon_\theta$ is a sufficiently small step size. By utilizing the updates in \eqref{FixedPointWithEqRateConstraint} and \eqref{eq:Lagrange}, the convergence of the fixed-point equation is proved as in Lemma \ref{FixedPointConvergence_noRateConstraint} by following the standard interference function in \cite{Yates1995}.
    \begin{lemma}\label{FixedPointConvergence_noRateConstraint} $\textbf{I}(\textbf{\textrm{p}})$ given in \eqref{PointWiseVector} with its elements defined in \eqref{FixedPointEq1:UnsatisfiedIoT devices} and \eqref{FixedPointEq1:SatisfiedIoT devices} is a standard interference function (SIF). As a consequence, starting from a feasible point and iteratively updating the transmit power coefficients based on \eqref{FixedPointWithEqRateConstraint}  converges to a fixed point that optimizes the objective of problem~\eqref{MaxRwithEqualRconstraint_v5}.
	\end{lemma}
	\begin{proof}
	    The proof is provided in Appendix \ref{Proof:FixedPointConvergence_noRateConstraint}
	\end{proof}
    \begin{algorithm}[t]
		\caption{Dual interference function-based method for power allocation with equality of service constraints (DIF-PA).} \label{LogApproximationAlgorithm_EqRateConstraint} 
		\begin{algorithmic}[1] 
		      \STATE \textbf{Input}: The system parameters $\mathcal{N}$, $\mathcal{K}$, $B$, $\sigma$, $P^{\max}$, the channel state information $\textbf{h}_n, \forall n\in \mathcal{N}$, the AP association $\boldsymbol{\Psi}^{*,(t-1)}$, the power vector $\textbf{\textrm{p}}[0]=\textbf{\textrm{p}}^{*,(t-1)}$, and the tolerance $\epsilon_1$.
			\STATE Set the iteration index $i=0$. Calculate $\boldsymbol{\alpha}(\textbf{\textrm{p}}[0])$ and $\boldsymbol{\beta}(\textbf{\textrm{p}}[0])$. 
			\REPEAT 
			\STATE {$i\gets i+1$. Set $s=0$.} 
			\REPEAT
            \STATE {Calculate $\textbf{\textrm{p}}(s+1)=\textbf{I}(\textbf{\textrm{p}}(s))$ as in \eqref{FixedPointWithEqRateConstraint}.}
   %          \FOR {$n=1:|\mathcal{N}|$}
			% \IF {$n\in \mathcal{Q}$}
			% \STATE {Update power following \eqref{FixedPointEq1:UnsatisfiedIoT devices}.}
   %          \ELSE
   %          \STATE {Update power following \eqref{FixedPointEq1:SatisfiedIoT devices}.}
			% \ENDIF
   %          \ENDFOR 
			\STATE Update Lagrangian multipliers $\theta_{k}{(s+1)}=\left \lceil \theta_{k}{(s)}+\epsilon_\theta\left(\sum_{n=1}^N\mu_{k,n}\exp(\bar{P}_n)-P_k^{\max}\right) \right \rceil^{+}$.
			\STATE {$s\gets s+1$.}
			\UNTIL {Convergence.}
			\STATE {Update $\boldsymbol{\alpha}(\textbf{\textrm{p}}^*[i])$ and $\boldsymbol{\beta}(\textbf{\textrm{p}}^*[i])$ using \eqref{LogApproximation1}.}
			\UNTIL {$\|\textbf{\textrm{p}}^*[i]-\textbf{\textrm{p}}^*[i-1]\|\leq \epsilon_1$.} 
			\STATE \textbf{Output}: The power allocation vector $\textrm{\textbf{p}}^{*,(t)}$. 
		\end{algorithmic}
    \end{algorithm} 
    After obtaining the optimized transmit powers for all the IoT devices, they are exploited to update the log approximation factors as in \eqref{LogApproximation1}. The details of the proposed optimization strategy are summarized in Algorithm \ref{LogApproximationAlgorithm_EqRateConstraint}. We note that  Algorithm \ref{LogApproximationAlgorithm_EqRateConstraint} will monotonically increase the objective function of \eqref{MaxRwithoutRconstraint_v3}, $R_{\textrm{tot}}(\textbf{\textrm{p}})$, which always converges to a fixed-point solution by utilizing the same methodology as shown in Lemma~\ref{LogApproximation1}.
    \subsubsection{AP association}\label{Subsection:AP_Selection}
    Given the power allocation, we reformulate the problem with respect to $\boldsymbol{\Psi}$ as below
	\begin{subequations} \label{MaxRwithoutRconstraint_APselection}
		\begin{align}
			&& \underset{\substack{ \boldsymbol{\Psi}}}{\textrm{maximize} } & \quad  R_{\textrm{tot}}(\boldsymbol{\Psi}),\\
			&&\textrm{subject to}\hspace{0.15cm}& \quad\underset{n\in \mathcal{N}}{\sum} \mu_{k,n}P_{n}\leq  P_{k}^{\textrm{max}},  \forall k\in \mathcal{K},\label{MaxRwithoutRconstraint_APselection: V1b} \\
            &&& \quad\sum_{k\in \mathcal{K}}\mu_{k,n}=1, \forall n\in \mathcal{N}, \label{MaxRwithoutRconstraint_APselection: V1c} \\
			&&&\quad R_n(\boldsymbol{\Psi},\textbf{\textrm{p}})= \xi_n^{\textrm{thr}}, \forall n\in \mathcal{Q}^* \label{MaxRwithoutRconstraint_APselection: V1d} 
		\end{align}
	\end{subequations} 
	The constraint in \eqref{MaxRwithoutRconstraint_v1:EqualRConstraint} is extremely difficult to satisfy, and thus the feasible set of the problem in \eqref{MaxRwithoutRconstraint_APselection} is almost empty. Therefore, we change the target of the AP association problem to optimize the total data throughput while guaranteeing the minimum data throughput for all IoT devices in the given satisfied IoT devices set. We rewrite the above optimization problem while keeping the total power constraint, \eqref{MaxRwithoutRconstraint_APselection: V1b}, and the AP-IoT devices association constraint, \eqref{MaxRwithoutRconstraint_APselection: V1c}, as follows:
    \begin{equation} \label{MaxRwithoutRconstraint_APselection2}
		\begin{aligned}
			&& \underset{\substack{ \boldsymbol{\Psi}}}{\textrm{maximize}} & \quad  R_{\textrm{tot}}(\boldsymbol{\Psi})\\
			&&\textrm{subject to}\hspace{0.15cm}& \quad\eqref{MaxRwithoutRconstraint_APselection: V1b},\eqref{MaxRwithoutRconstraint_APselection: V1c}, \\
			&&&\quad R_n(\boldsymbol{\Psi},\textbf{\textrm{p}})\geq R_n^{\textrm{thr}}, \forall n\in \mathcal{Q}^{*}. 
   % \label{MaxRwithoutRconstraint_APselection:EqualRConstraint} 
		\end{aligned}
    \end{equation} 
    Our approach is to apply the coalition game method to search an AP association structure, which aims to provide a better total data throughput to all IoT devices. Therein, IoT devices in the system play the role of players, while the change of AP association from each IoT device is considered as a strategy. We define the set of IoT devices being served by the $k$-th AP as $\mathcal{N}_k= \{n|\mu_{k,n}=1, n\in \mathcal{N}\}$. The current AP association structure of all IoT devices is then defined as $\mathcal{F}^{\textrm{curr}}=\{\mathcal{N}_1,..., \mathcal{N}_K\}$. We now have two conditions $\mathcal{N}_k\cup \mathcal{N}_j =\emptyset, \forall k,j \in \mathcal{K}$ and $\underset{k\in\mathcal{K}}{\cup} \mathcal{N}_k = \mathcal{N}$. If the $n$-th IoT device change the AP association from $k$-th to $j$-th, a new AP association structure is formed $\mathcal{F}^{\textrm{temp}}=(\mathcal{F}^{\textrm{curr}}\backslash {\mathcal{N}_k,\mathcal{N}_j})\cup (\mathcal{N}_k\backslash n)\cup (\mathcal{N}_j\cup n)$. The new AP association structure is accepted if it satisfies these two following updated policies
    \begin{align}
        &R_n(\mathcal{F}^{\textrm{temp}})\geq \xi_n^{\textrm{thr}}, \forall n\in \mathcal{Q}^*, \label{CGcondition1}\\
        &R_{\textrm{tot}}(\mathcal{F}^{\textrm{temp}})> R_{\textrm{tot}}(\mathcal{F}^{\textrm{curr}}),\label{CGcondition2}
    \end{align}
    where \eqref{CGcondition1} ensures that IoT devices belonging to the satisfied IoT devices set $\mathcal{Q}^*$ still meet the required data throughput. Meanwhile, \eqref{CGcondition2} ensures that we only can change the AP association of unsatisfied IoT devices when the total data throughput is improved. The expectation of the final AP association structure is that all IoT devices agree that this structure is the best from every IoT devices' view. Intuitively, no IoT device can find a better strategy to improve the player's utility. Given the strategies, updated policies, and the set-up target, we repeatedly traverse through all IoT devices in the system and change their AP association until obtaining the final AP association structure. We introduce the variable $\iota$ to stop the loop. In particular, we set $\iota=0$ at the beginning of every loop, and the algorithm runs through all IoT devices to testify whether they can find a new AP that achieves a higher total network throughput. If $\iota=0$ after the algorithm traverses to all IoT devices, no IoT device wants to change its AP association, and the algorithm stops. Whereas $\iota=1$ means that at least one IoT device changes the AP association and forms a new AP association structure in the system, therefore, the algorithm sets $\iota=0$ and testifies all IoT devices again. The detailed method is provided in Algorithm \ref{AlgorithmCG}. Note that the number of APs and IoT devices is limited, thus, the number of AP association structure candidates is finite. Therefore, Algorithm \ref{AlgorithmCG} always stops at a certain iteration.
    \begin{algorithm}[t]
		\caption{Coalition game-based AP association (CG-APA)} \label{AlgorithmCG}
		\begin{algorithmic}[1]
			\STATE \textbf{Input}:  The system parameters $\mathcal{N}$, $\mathcal{K}$, $B$, $\sigma$, $P^{\max}$, the channel state information $\textbf{h}_n, \forall n\in \mathcal{N}$, the power vector $\textbf{\textrm{p}}^{*,(t-1)}$.
            \STATE \textbf{Initialized:} $\mathcal{F}^{\textrm{curr}}$ based on the AP association $\boldsymbol{\Psi}^{*,(t-1)}$.
% 			\STATE Set the current AP association as $\boldsymbol{\Psi}^{\textrm{curr}}$.
			\REPEAT 
			\STATE {Set $\iota=0$.}
			\FOR {$n=1: N$}
			\FOR {$k=1: K$}
			\STATE {Assume that $n\in \mathcal{N}_j$ then $\mathcal{F}^{\textrm{temp}}=(\mathcal{F}^{\textrm{curr}}\backslash {\mathcal{N}_k,\mathcal{N}_j})\cup (\mathcal{N}_k\backslash n)\cup (\mathcal{N}_j\cup n)$.}
			\STATE {Compute $R_n(\mathcal{F}^{\textrm{curr}})$ and $R_n(\mathcal{F}^{\textrm{temp}}), \forall n\in \mathcal{N}$.}
			\IF {\eqref{CGcondition1} is violated}
			\STATE {\textbf{Continue}.}
			\ELSIF {\eqref{CGcondition1} is satisfied}
			\STATE {Set $\mathcal{F}^{\textrm{temp}}=\mathcal{F}^{\textrm{curr}}$.}
			\STATE {$\iota=1$.}
			\ENDIF 
			\ENDFOR
			\ENDFOR
			\UNTIL {$\iota=0$.}
			\STATE \textbf{Output}: The AP association $\boldsymbol{\Psi}^{*,(t)}$.
		\end{algorithmic}
    \end{algorithm}
 
%     \begin{lemma}\label{CG_APselection}
% 	    Given any AP selection structure and the power allocation of all IoT devices, the proposed coalition game-based method will always converge to the final coalition structure, after a finite number of iterations.
%     \end{lemma}
% Note that the number of APs and the number of IoT devices are limited, and thus the total number of AP selection structure candidates is finite. Therefore, the algorithm always stops at a certain iteration.
    \subsection{Alternative Algorithm}
    
    In order to obtain the final solution, an alternating algorithm is proposed. To start the algorithm, we set up the initial AP-IoT devices association, $\boldsymbol{\Psi}^{*,(0)}$, and the power allocation, $\textbf{\textrm{p}}^{*,(0)}$, as in Subsection \ref{ProposedSolution}. The power allocation problem with fixed rate constraints is resolved first in every loop as in Subsection \ref{Subsection:PowerAllocation}. Next, the AP-IoT devices association is obtained via the coalition game-based method as in Subsection \ref{Subsection:AP_Selection}. The procedure is terminated after checking two following conditions in order: the improvement in the number of satisfied IoT devices is first checked and the improvement in the total network throughput is the second condition. If the number of satisfied IoT devices cannot be increased, the procedure is stopped when the total data throughput cannot increase higher than a sufficient small threshold, say $\epsilon_2$. The detailed procedure is summarized as in Algorithm \ref{AAlgorithm}. In the following Lemma, the convergence of Algorithm \ref{AAlgorithm} will be proved. Convergence speed and running time will be given in Section \ref{ExperimentalEvaluation} to evaluate the computational requirements of Algorithm \ref{AAlgorithm} for the system.

    % After obtaining the solutions to the power allocation AP association problem as in Sections \ref{Subsection:PowerAllocation} and \ref{Subsection:AP_Selection},  one observes that the output power vector of Algorithm \ref{LogApproximationAlgorithm_EqRateConstraint} always improves or at least keeps the number of satisfied IoT devices remains unchanged. Meanwhile, the coalition game-based method in Algorithm~\ref{AlgorithmCG} generates a new AP association structure that guarantees the improvement of the total data throughput. Moreover, the policies of Algorithm~\ref{AlgorithmCG} ensure that the number of satisfied IoT devices does not decrease. Therefore, by alternating solving the power allocation and the AP association, we can gradually increase both the number of satisfied IoT devices and the total data throughput. In order to terminate the alternative procedure, we verify two conditions in order: the improvement in the number of satisfied IoT devices is first checked and the improvement in the total network throughput is the second condition. If the number of satisfied IoT devices cannot be increased, the procedure is terminated when the total data throughput cannot increase higher than a threshold, say $\epsilon_2$. The detailed procedure is summarized as in Algorithm \ref{AAlgorithm}. In the following Lemma, the convergence of Algorithm \ref{AAlgorithm} will be proved.
    \begin{algorithm}[t]
		\caption{Alternative algorithm to solve the problem in \eqref{DualProblem_v1} (AA Algorithm)} \label{AAlgorithm} 
		\begin{algorithmic}[1] 
		  \STATE \textbf{Input}:  The system parameters $\mathcal{N}$, $\mathcal{K}$, $B$, $\sigma$, $P^{\max}$, the channel state information $\textbf{h}_n, \forall n\in \mathcal{N}$, and the tolerance $\epsilon_2$.
            \STATE \textbf{Initialize} the AP association $\boldsymbol{\Psi}^{*,(0)}$ as mention early in Subsection \ref{ProposedSolution}. 
            \STATE Calculate $\textbf{\textrm{p}}^{*,(0)}$ based on Algorithm \ref{LogApproximation1}.
            \STATE Calculate $R_n(\boldsymbol{\Psi}^{*,(0)},\textbf{\textrm{p}}^{*,(0)}), \forall n\in \mathcal{N}$ and $R_{\textrm{tot}}((\boldsymbol{\Psi}^{*},\textbf{\textrm{p}}^{*,(0)}))$.
			\STATE Determine the set of satisfied IoT devices $\mathcal{Q}^{*}(\boldsymbol{\Psi}^{*,(0)},\textbf{\textrm{p}}^{*,(0)})$ based on \eqref{QsetDeterminantEq}. 
			\REPEAT  
			\STATE {$t\gets t+1$.} 
			\STATE {Solve the problem \eqref{MaxRwithoutRconstraint_P_v1} to obtain $\textbf{\textrm{p}}^{*,(t)}$ given $\boldsymbol{\Psi}^{*,(t-1)}$.} 
			\STATE {Solve the problem \eqref{MaxRwithoutRconstraint_APselection} to obtain $\boldsymbol{\Psi}^{*,(t)}$ given $\textbf{\textrm{p}}^{*,(t-1)}$.} 
            \STATE {Calculate $R_n(\boldsymbol{\Psi}^{*,(t)},\boldsymbol{\Psi}^{*,(t)}), \forall n\in \mathcal{N}$ and $R_{\textrm{tot}}(\boldsymbol{\Psi}^{(t)},\textbf{\textrm{p}}^{(t)})$.}
			\STATE {Update $\mathcal{Q}^{*,(t)}=\{n|R_n(\textbf{\textrm{p}}^{*,(t)})\geq R_n^{\textrm{thr}}, \forall n\in \mathcal{N}\}$.}
            \STATE Calculate 
            \begin{equation*}
                \begin{aligned}
                    \nabla = \frac{R_{\textrm{tot}}(\boldsymbol{\Psi}^{*,(t)},\boldsymbol{\Psi}^{*,(t)})-R_{\textrm{tot}}(\boldsymbol{\Psi}^{*,(t-1)},\textbf{\textrm{p}}^{*,(t-1)})}{R_{\textrm{tot}}(\boldsymbol{\Psi}^{*,(t-1)},\textbf{\textrm{p}}^{*,(t-1)})}.
                \end{aligned}
            \end{equation*}
			\UNTIL {$(|\mathcal{Q}^{*,(t)}|==|\mathcal{Q}^{*,(t-1)}|)\mbox{ and }(\nabla\leq \epsilon_2)$.} 
			\STATE \textbf{Output}: $\boldsymbol{\Psi}^{*,\textbf{fin}}$,$\textbf{\textrm{p}}^{*,\textbf{fin}}$. 
		\end{algorithmic}
            
    \end{algorithm} 
    
    \begin{lemma}\label{Lemma: Convergence}
        {Given the limited number of APs, $K$, and the limited power budget at APs, $P_k^{\max}, \forall k$, the optimization problem in \eqref{Problemv1} is considered to be infeasible. Therefore, achieving the optimal or a local solution cannot be ensured. The problem in \eqref{Problemv1} is reformulated as the dual objective problem in \eqref{DualProblem_v1} in order to transform the problem in \eqref{Problemv1} into a feasibility and find a good solution to the system under infeasible circumstances. 
        
        Therefore, the problem in \eqref{DualProblem_v1} is solved by Algorithm \ref{AAlgorithm} that prioritizes providing services to as many IoT devices from the IoT devices set $\mathcal{N}$ as possible. Particularly, Algorithm \ref{AAlgorithm} gradually finds and adds potential IoT devices into the set of satisfied IoT devices, $\mathcal{Q}$, by repeatedly reallocating the radio resources via solving the problem in \eqref{MaxRwithoutRconstraint_v1}. Algorithm \ref{AAlgorithm} ensures the non-decrease of the number of satisfied IoT devices. When the number of satisfied IoT devices is fixed, Algorithm \ref{AAlgorithm} also guarantees the non-decrease of the total network throughput and always converges.} 
    \end{lemma}
    \begin{proof}
        The proof is given in Appendix \ref{Proof: Convergence}.
    \end{proof}

	\subsection{Complexity Analysis}
	The optimal solution of the problem in \eqref{Infeasible checker 1} can be achieved via a matrix-inversion computation. The complexity of the modified BB algorithm is in the order of 
	%\begin{equation*}
		$\mathcal{O}(N_1N^2)$,
	%\end{equation*}
	where $N_1$ is the number of iterations needed to achieve the convergence of the modified BB algorithm. Meanwhile, the computational complexity of AA Algorithm is in the order of 
	%\begin{equation*}
		$\mathcal{O}(N_2N_3KN^3 + N_2N_4 KN^2)$,
	%\end{equation*} 
	where $N_2$ and $N_3$ are the numbers of iterations for convergence of AA algorithm and DIF-PA algorithm, respectively. Besides, $N_4$ is the maximum number of iterations for the convergence of the coalition game-based AP association algorithm. In the worst case, both the modified BB algorithm and CG-APA algorithm experience the complexity as the exhaustive search with $N^K$ possible of AP-IoT devices association candidates. Due to the consideration of infeasible circumstances, the complexity of the modified BB algorithm can be reduced significantly by pruning many branches. The use of AA algorithm consisting of DIF-PA and CG-APA algorithms requests a higher order of complexity in exchange for a better network throughput. However, it should be noted that the practical running time depends on the network size and specific convergence thresholds, and thus, the time complexity will be evaluated in the next section.
	
    \section{Performance Evaluation}\label{ExperimentalEvaluation}
    In this section, the performance of our proposed algorithms is evaluated via extensive numerical results together with comparisons with several benchmarks. In particular, the proposed AP association is compared with the nearest AP association method. Meanwhile, the performance of the dual interference function-based method is compared with an equal power allocation scheme, where all IoT devices are allocated by a power level. Note that we conducted simulations on a personal computer equipped with the following configuration: [AMD Ryzen $5$ $5600G$ with Radeon Graphics $3.90$ GHz and an installed RAM of $16.0$ GB].
    \subsection{Channel Model and System Parameters}
    Even though the radio environment varies over time and frequency, we adopt the quasi-static model where all the propagation channels follow uncorrelated Rayleigh distribution patterns, which are formulated as \cite{Chien2016},
    \begin{equation}
        h_{k,n}\sim \mathcal{CN} (0,\vartheta_{k,n}), \forall k\in \mathcal{K}, n\in \mathcal{N},
    \end{equation}
    where $\vartheta_{k,n}$ is the large-scale fading coefficient, which is derived from the path loss and the shadowing fading. Specifically, we formulate $\vartheta_{k,n}$ as follows \cite{Chien2016}: 
    \begin{equation}
        \vartheta_{k,n}=\zeta_{k,n}10^{\frac{\psi \sigma_{k,n}}{10}},
    \end{equation}
    where $\psi=7$ is the standard deviation of shadowing and $\sigma_{k,n}$ is distributed as $\mathcal{CN}(0,1)$. Meanwhile, the path loss coefficient $\zeta_{k,n}$ is defined as, \cite{Chien2016},
    \begin{equation}
        \zeta_{n,k}\hspace{0.2cm} [\textrm{dB}]=-(120.9+37.6\log_{10}(d_{n,k} )),
    \end{equation}
    where $d_{n,k}$ is the distance between the $n$-th IoT device and the $k$-th AP. We consider a circle communication area with a radius of $r$ meters. Inside this area, there are $N$ IoT devices and $K$ APs. The distance between two APs should exceed $30$~[m] to avoid multiple APs locating at the same area. We illustrate an example of our network in Fig. \ref{coordinate}, where the random location of APs and IoT devices are shown. The common network parameters are selected based on the 3gpp standard of NB-IoT \cite{3gppStandardizationNBIOT} and are listed in Table~\ref{Table1}.
    \begin{figure}[t]
		\includegraphics[trim=3.4cm 8.5cm 1.5cm 8.5cm, clip=true, width=4.0in]{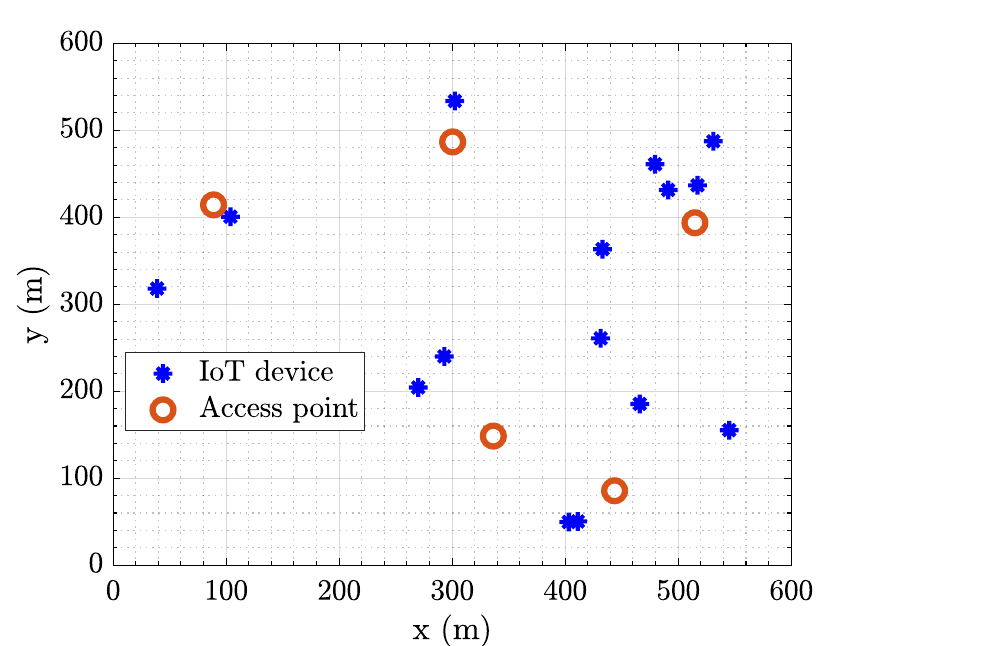}
		\caption{An example coordinates of $5$ APs and $15$ IoT devices.}
		\label{coordinate}
	\end{figure}
    \begin{table}[t] 
    	\caption{System parameters \cite{3gppStandardizationNBIOT}.}
    	\centering
    	\begin{tabular}{l|c}
%    		\toprule[1pt]
    		\midrule[0.3pt] 
    		Cell radius (r) & 300 [m]\\ 
    		\midrule 
    		Maximum transmit power of each AP ($P^{\max}$) & 23 [dBm]\\ 
    		%		\midrule 
    		%		Operating Frequency & 1.9 [GHz]\\ 
    		\midrule 
    		Bandwidth $B$& 180 [kHz]\\ 
    		\midrule  
    		Power spectral density of the thermal noise ($\sigma$) & -174 [dBm/Hz]\\ 
    		%		\midrule 
    		%		Maximum number of AP in the area & 10\\ 
    		\midrule 
    		Minimum distance between two APs & 30 [m]\\ 
    		\midrule 
    		Lagrangian multiplier step size ($\kappa$) & 0.0001\\ 
    		\midrule 
    		Terminated conditions ($\epsilon_1$ and $\epsilon_{2}$)  & 0.0001\\
    		\bottomrule[1pt] 
    	\end{tabular}
    	\label{Table1}
    \end{table}% 
    \subsection{Numerical Results}
    We first introduce two following benchmarks: $i)$ \textit{Nearest AP association (Nearest-APA) \cite{Chien2016}}:  Each IoT device selects the closest AP in terms of geometry. In this way, each IoT device will have the channel gain condition with the lowest path loss; $ii)$ \textit{Equal power allocation (Equal-PA) \cite{Adu2019}}: The transmit power is fixed, and all are equal to each other. To demonstrate the efficiency of proposed algorithms, the following methods will be compared:
    \begin{itemize}
    \item[$i)$] Brute Force: The system applies brute force to find the solution that serves the most number of IoT devices.
    \item[$ii)$] The modified BB algorithm: The system applies the modified BB algorithm to achieve both power allocation and association with AP-IoT devices.
    \item[$iii)$] Proposed DIF-PA + Proposed CG-APA: The system applies the AA algorithm, where the proposed DIF-PA and the proposed CG-APA are used to achieve the power allocation and the AP-IoT devices association, respectively. 
    \item[$iv)$] Proposed DIF-PA + Nearest-APA: The system applies AA algorithm, where the proposed DIF-PA and the nearest CG-APA are used to achieve the power allocation and the AP-IoT devices association, respectively.  
    \item[$v)$] Equal-PA + CG-APA: The system applies AA algorithm, where Equal-PA and CG-APA are used to achieve the power allocation and the AP-IoT devices association, respectively.  
    \item[$vi)$] Equal-PA + Nearest-APA: The system applies AA algorithm, where the proposed DIF-PA and the nearest CG-APA are used to achieve the power allocation and the AP-IoT devices association, respectively.  
    \end{itemize}
    \begin{figure}[t]
    		\includegraphics[trim=3.7cm 8.5cm 3.0cm 9cm, clip=true, width=3.8in]{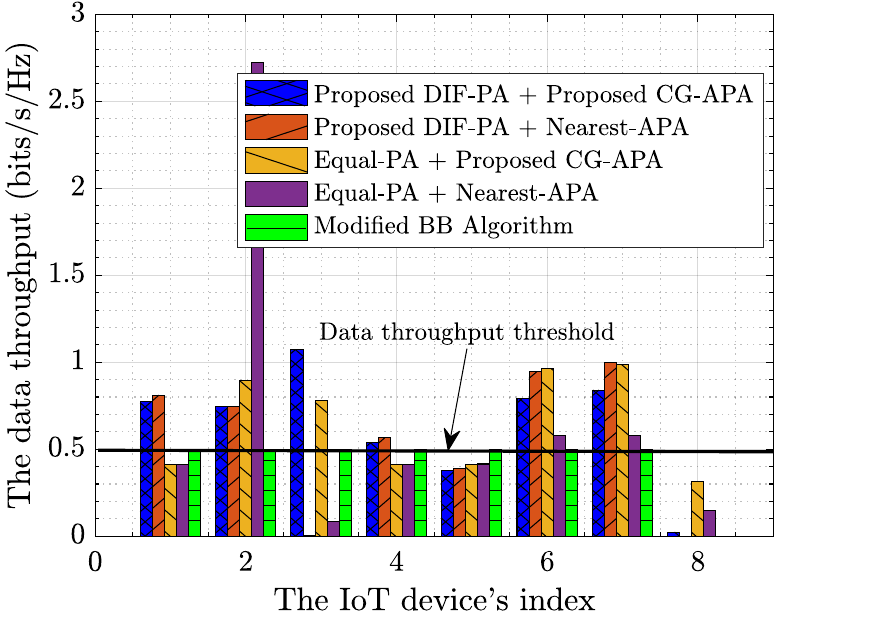}
    		\caption{The achievable rate of IoT IoT devices with different solutions when the system has 8 IoT IoT devices and 3 APs.}
    		\label{fig:UserRate}
            \vspace{-1.5em}
    \end{figure} 
    Fig.~\ref{fig:UserRate} illustrates the detailed data throughput of IoT devices when the system deploys different power allocation solutions and AP association methods. The network has eight IoT devices with the minimum data throughput requirement of $0.5$ (bits/s/Hz) from all IoT devices. When the system applies the modified BB algorithm, seven IoT devices are satisfied the exact requested data throughput while only the $8$-th device is out of service. On the other hand, six IoT devices can satisfy their request when the system combines the proposed DIF-PA and CG-APA. Specifically, one IoT device obtains a data throughput of $1.1$ (bits/s/Hz), and three IoT devices with data throughput higher than $0.7$ (bits/s/Hz). The numerical result implies that AA Algorithm simultaneously optimizes both the number of satisfied IoT devices and the total network throughput. Therein, the total network throughput is around $5.5$ (bits/s/Hz), which outperforms the system applying the modified BB algorithm with a total network throughput of $3.5$ (bits/s/Hz) only. Meanwhile, the benchmark that involves a combination of the proposed DIF-PA and Nearest-APA can serve five IoT devices with a total data throughput of $4.4$ (bits/s/Hz). Overall, the proposed CG-APA can consistently achieve greater performance than Nearest-APA. When the system applies Equal-PA and the proposed CG-APA, only four IoT devices can meet their requirements. Lastly, three IoT devices can meet the requested data throughput, and one obtains a very high data throughput of $2.7$ (bits/s/Hz). However, three IoT devices obtained a data throughput of $0.4$ (bits/s/Hz), in which two out of three IoT devices improved their data throughput over the requirement when the system used the proposed DIF-PA and the proposed CG-APA.

    \begin{figure}
    		\includegraphics[trim=3.6cm 8.5cm 3.5cm 9.0cm, clip=true, width=3.6in]{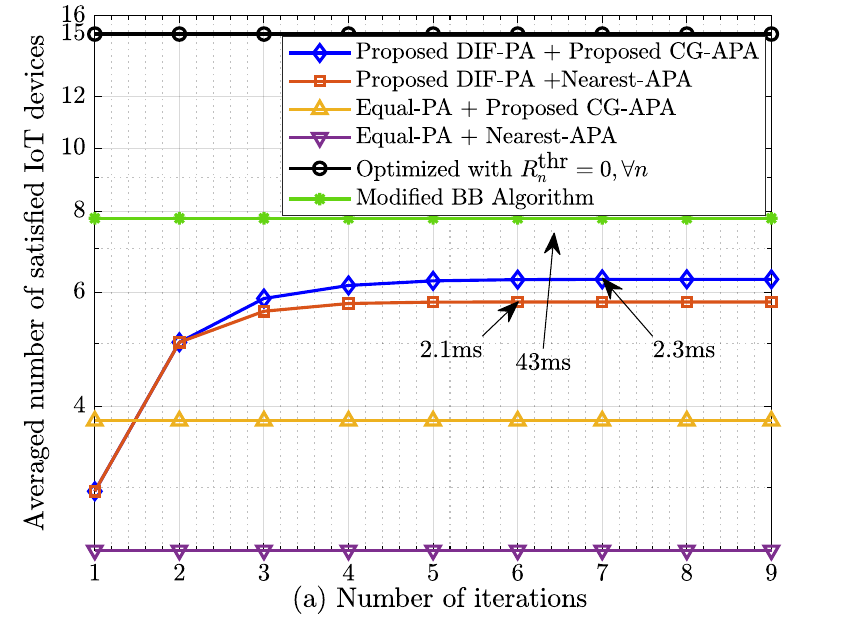}
    		\caption{The convergence in the average number of satisfied IoT devices when the system serves 15 IoT devices with the minimum data throughput $0.5$ (bits/s/Hz) by 5 APs.}
    		\label{fig:Convergence_vs_Iterations_NoSatisfiedIDs}
            \vspace{-1.5em}
    \end{figure} 
    
    \begin{figure}
    		\includegraphics[trim=3.6cm 8.5cm 3.5cm 9.0cm, clip=true, width=3.6in]{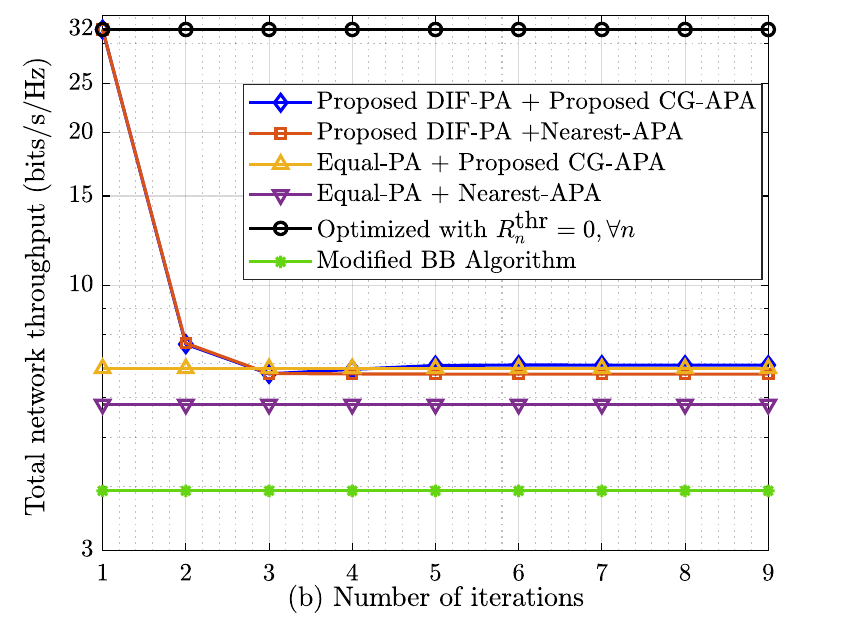}
    		\caption{The convergence in the total network throughput when the system serves 15 IoT devices with the minimum data throughput $0.5$ (bits/s/Hz) by 5 APs.}
    		\label{fig:Convergence_vs_Iterations_TotalRate}
            \vspace{-1.5em}
    \end{figure} 
    Figs.~\ref{fig:Convergence_vs_Iterations_NoSatisfiedIDs} and \ref{fig:Convergence_vs_Iterations_TotalRate} illustrate the convergence regarding the average number of satisfied IoT devices and the total network throughput, respectively. We can observe the first underlying cost of the modified BB algorithm and the proposed DIF-PA from Fig. \ref{fig:Convergence_vs_Iterations_TotalRate} is that the total network throughput experiences a significant reduction in exchange for an increase in the number of satisfied IoT devices. This disadvantage can be seen clearly when the system applies the modified BB algorithm, where the system can satisfy approximately $7.8$ IoT devices with a running time of $43$ms, however, the total data network throughput is only $3.9$ (bits/s/Hz). Meanwhile, the system deploying the proposed DIF-PA and the proposed CG-APA converges in approximately $7$ iterations on average, with a time consumption of around $2.3$ms, resulting in a total network throughput of approximately $7.0$ (bits/s/Hz), which is superior to applying the modified BB algorithm. Alternatively, by combining the proposed DIF-PA with Nearest-APA, the system converges quicker after $6$ iterations with a time consumption of approximately $2.1$ms and a total network throughput of $6.7$ (bits/s/Hz). Upon reaching convergence, the combination of the proposed DIF-PA and CG-APA manages to serve an average of $6.3$ satisfied IoT devices. This performance slightly surpasses that of the proposed DIF-PA coupled with Nearest-APA, which yields an average of approximately $5.8$ satisfied IoT devices. In contrast, when the system focuses solely on AP-IoT device associations, the complexity and calculation requirements reduce noticeably, resulting in minimal time consumption. However, this approach significantly reduces performance. Particularly, the system deploying Equal-PA with the proposed CG-APA method can serve $3.9$ IoT devices with a total network throughput of $6.8$ (bits/s/Hz). The combination of Equal-PA and Nearest-APA performs the poorest, providing services to only $2.4$ IoT devices.
    
    \begin{figure}
    		\includegraphics[trim=3.6cm 8.5cm 3.5cm 9.0cm, clip=true, width=3.5in]{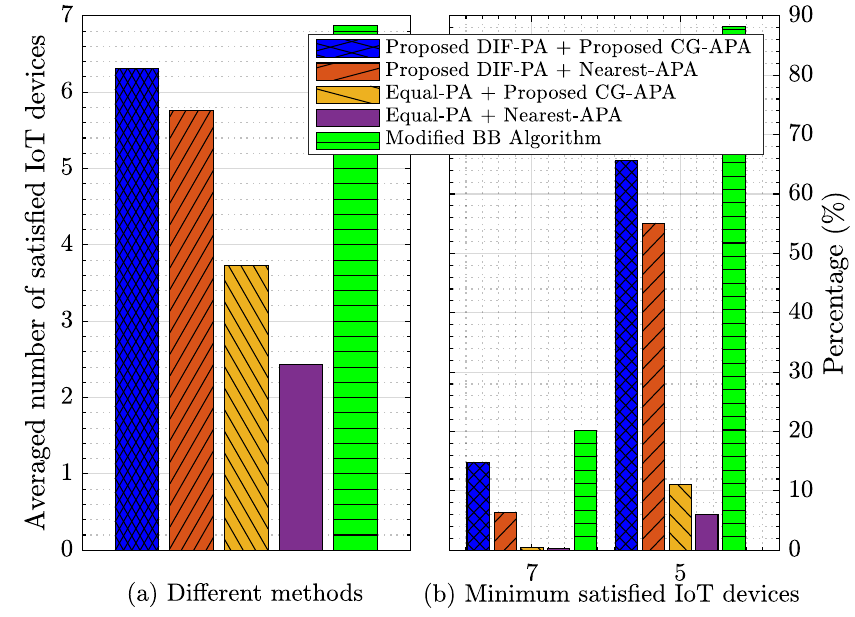}
    		\caption{The left figure shows the averaged number of satisfied IoT devices and the right figure shows the probability that the system can serve at least 5 or 7 IoT devices at least $0.5$ (bits/s/Hz) when the system has 15 IoT devices and 3 APs.}
    		\label{fig:PercentageSatisfiedDevices}
            \vspace{-1.5em}
    \end{figure} 
    
    Fig.~\ref{fig:PercentageSatisfiedDevices}(a) shows the average number of satisfied IoT devices versus benchmarks that are averaged over different $10,000$ realizations of user locations and shadow fading. Meanwhile, Fig.~\ref{fig:PercentageSatisfiedDevices}(b) illustrates the percentage that the system serves at least a certain number of IoT devices, i.e., $5$ or $7$ IoT devices. More specifically, Fig.~\ref{fig:PercentageSatisfiedDevices}(a) displays the superior performance of the modified BB algorithm and the proposed DIF-PA compared to Equal-PA. Approximately $6.8$ IoT devices can be satisfied when the system applies the modified BB algorithm. When the system applies the proposed DIF-PA, approximately $5.7$-$6.3$ IoT devices are served, where the exact number depends on the selected AP association method. Specifically, combining the proposed DIF-PA and CG-APA can serve an average of $6.3$ IoT devices. However, the proposed DIF-PA and Nearest-APA can only satisfy $5.8$ IoT devices on average. In contrast, the Equal-PA only provides service to around $2.4$-$3.8$ IoT devices. In Fig.~\ref{fig:PercentageSatisfiedDevices}(b), $89$\% of the realizations, the modified BB algorithm can serve at least $5$ IoT devices. When the system shares the radio resource via optimizing the dual-objective function, we observe that $66$\% of the realizations, the system can serve at least $5$ IoT devices when utilizing the combination of the proposed DIF-PA and CG-APA. Whereas the combination of the proposed DIF-PA and Nearest-APA only achieves $55$\%. Furthermore, utilizing Equal-PA and the proposed CG-APA can serve more than $5$ IoT devices by their requirements with the probability of less than $0.1$. With more than $7$ IoT devices in the network, Equal-PA with any AP association method cannot serve any IoT device to meet their requirement. The combination of the proposed DIF-PA and CG-APA can help at least $7$ IoT devices achieve a data throughput higher than the threshold with the probability of approximately $0.16$. Meanwhile, by exploiting the proposed DIF-PA and Nearest-APA, the system provides service for more than $7$ IoT devices with the probability of only about $0.07$. Finally, with the modified BB algorithm, the probability of serving at least $7$ IoT devices is approximately $20.5$\%.
    
    \begin{figure}
    		\includegraphics[trim=3.7cm 8.2cm 3.5cm 9cm, clip=true, width=3.5in]{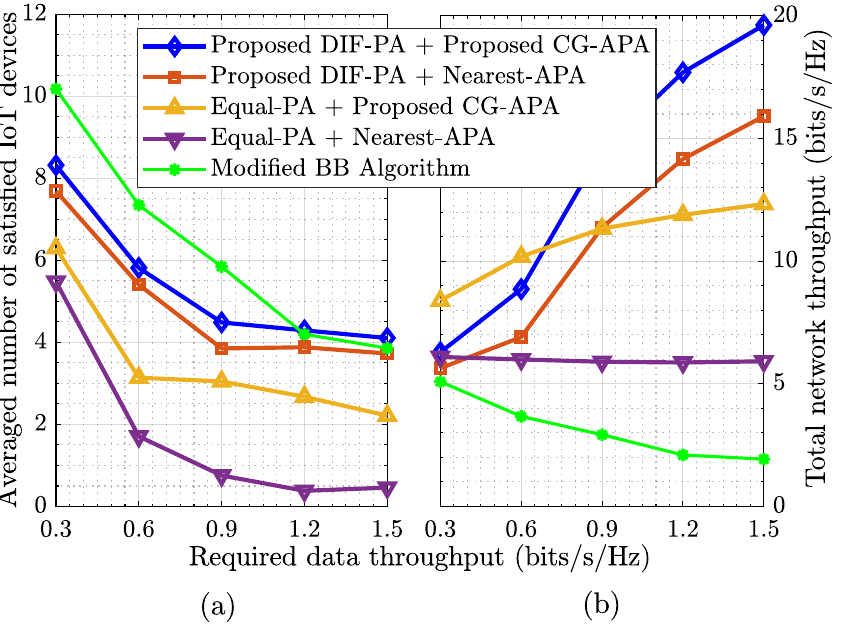}
    		\caption{The number of satisfied IoT devices and the total data throughput versus the change in the required data throughput when the system has 15 IoT devices and 5 APs.}
    		\label{fig:AverageSatisfiedIoT device_n_TotalRate_vs_RateRequirement}
            \vspace{-1.5em}
    \end{figure} 
    % \begin{figure}
    % 		\includegraphics[trim=3.9cm 8.5cm 3.5cm 9cm, clip=true, width=3.8in]{fic/TotalRate_vs_RateRequirement.pdf}
    % 		\caption{The total achievable rate of all IoT devices versus the change in the required data throughput when the system has 15 IoT devices and 5 APs.}
    % 		\label{fig:TotalRate_vs_RateRequirement}
    % \end{figure}
    Fig.~\ref{fig:AverageSatisfiedIoT device_n_TotalRate_vs_RateRequirement} illustrates network performance as the required data rate varies. It highlights the efficiency of optimizing just the number of satisfied IoT devices versus optimizing both the number of satisfied devices and total network throughput. Using the modified BB algorithm, approximately $10.1$ IoT devices are served when the required data throughput is $0.3$ (bits/s/Hz). This number decreases as the required data throughput increases, dropping to $3.95$ IoT devices at $1.5$ (bits/s/Hz). The proposed DIF-PA method, particularly when combined with the CG-APA approach, shows superior performance compared to Equal-PA. Specifically, the DIF-PA method combined with CG-APA serves an average of $8.2$ IoT devices at $0.3$ (bits/s/Hz), while the Equal-PA method combined with CG-APA serves around $7.7$ devices, and the Equal-PA method combined with Nearest-APA serves about $6.2$ devices. Interestingly, as the required data throughput increases, the modified BB algorithm becomes less effective compared to the DIF-PA method combined with CG-APA, particularly under high data rate demands. This is because the modified BB algorithm adds users sequentially, which becomes inefficient when user requirements are too high. In contrast, the DIF-PA method selectively serves users that are easier to satisfy with higher data rates, making it more effective in challenging conditions. As the required data throughput increases, the number of satisfied devices decreases across all methods, though total network throughput generally increases, except for the Equal-PA method combined with Nearest-APA. At a requirement of $0.6$ (bits/s/Hz), the number of satisfied devices for the Equal-PA method drops rapidly to $2$-$3$ IoT devices, while the DIF-PA method shows a more gradual decline. At $1.5$ (bits/s/Hz), the DIF-PA method combined with CG-APA achieves a total network throughput of approximately $20$ (bits/s/Hz) while serving around $4.2$ devices, compared to $3.8$ devices with the DIF-PA method combined with Nearest-APA. Notably, CG-APA consistently yields about $8$\% more satisfied devices than Nearest-APA.
    
    % \begin{figure}
    % 		\includegraphics[trim=3.9cm 8.5cm 3.5cm 9cm, clip=true, width=3.8in]{fic/AverageSatisfiedIoT device_vs_NoIoT devices1.pdf}
    % 		\caption{The number of satisfied IoT devices versus the change in the number of IoT devices when the system has 5 APs and the required data throughput is $0.5$ (bits/s/Hz).}
    % 		\label{fig:ChangeNumberIoT devices}
    % \end{figure} 

    % \begin{figure}
    % 		\includegraphics[trim=3.6cm 8.3cm 3.5cm 9cm, clip=true, width=3.45in]{fic/AverageSatisfiedID_n_TotalRate_vs_NoIDs1.pdf}
    % 		\caption{The number of satisfied IoT devices and the total data throughput versus the change in the number of IoT devices when the system has 5 APs and the required data throughput is $0.5$ (bits/s/Hz).}
    % 		\label{fig:AverageSatisfiedIoT device_n_TotalRate_vs_NoIoT devices1} 
    % \end{figure} 
    \begin{table*}
        \caption{Number of satisfied IoT devices and the network sum rate with different network sizes. The left and right values show the number of satisfied IoT devices and the total network throughput, respectively.}
		\centering
		\begin{tabular}{|c|c|c|c|c|c|c|} 
			\toprule[1pt]\midrule[0.3pt] 
			K&N & \text{DIF-PA + CG-APA} & \text{DIF-PA + Nearest-APA} & $\text{Equal-PA + CG-APA}$ & \text{Equal-PA + Nearest-APA} & $\text{Modified BB Algorithm}$\\ 
			\midrule 
			3 & 15 & $4.7/5.3$ & $4.4/4.8$ & $0.8/3.5$ & $0.8/3.2$ & $6.7/3.4$ \\ 
			\midrule 
			4 & 15 & $5.5/6.3$ & $5.3/5.7$ & $1.9/5.1$ & $1.7/4.5$ & $7.4/3.7$ \\ 
			\midrule 
			5 & 15 & $6.3/7.1$ & $5.9/6.4$ & $2.7/6.6$ & $2.5/5.7$ & $7.8/3.9$ \\ 
			\midrule 
			5 & 12 & $5.9/6.9$ & $5.5/6.7$ & $3.5/7.1$ & $3/5.8$ & $7.8/3.9$ \\ 
			\midrule 
			5 & 18 & $6.3/7.2$ & $6.1/6.9$ & $2.4/6.3$ & $2.1/5.4$ & $7.8/3.9$ \\ 
			\midrule 
			5 & 24 & $6.6/7.5$ & $6.4/7.2$ & $1.5/5.6$ & $1.5/5.1$ & $7.8/3.9$ \\ 
			\bottomrule[1pt] 
		\end{tabular}
		\label{Table: NoSatisfiedIDs_vs_APs_IDs}
		\vspace{-5pt}
	\end{table*}
    Table.~\ref{Table: NoSatisfiedIDs_vs_APs_IDs} shows the system performance versus different network settings when the requested data rate of all IoT devices is $0.5$ (bits/s/Hz). It is obvious that both the number of satisfied IoT devices and the total network throughput increase as more APs are added to the network. When the system applies the proposed DIF-PA and the CG-APA, the number of satisfied IoT devices increases from $4.7$ to $6.2$ when the number of APs increases from $3$ and $5$, respectively. These numbers of the proposed methods are smaller than the modified BB algorithm, around $15-20$\% at all settings. However, the proposed methods outperform the modified BB algorithm in terms of total network throughput. While the system combining DIF-PA and CG-APA serves $6.2$ IoT devices with a total of $7.1$ (bits/s/Hz), that of the modified BB algorithm is $7.8$ IoT devices but with only $3.9$ (bits/s/Hz). In comparison with different power allocation methods, DIF-PA presents a significant gap compared to Equal-PA. In detail, applying the Equal-PA and the CG-APA, the number of satisfied IoT devices is only around one-third of the combination of the DIF-PA and the CG-APA with the same set-up. Around $0.8$ and $2.7$ IoT devices are served the requested data rate when the system has $3$ and $5$ APs, respectively. Meanwhile, in terms of IoT devices-AP association, the proposed CG-APA also shows an improvement compared to the nearest-APA. When the number of APs increases from $3$ to $5$, the system applying DIF-PA and CG-APA can serve $4.7$ to $6.2$ IoT devices, while combining DIF-PA with Nearest-APA, the system serves $4.4$ to $5.9$ IoT devices only. When the number of IoT devices joining the network increases, only the system that uses the proposed DIF-PA presents an increasing trend in both the number of satisfied IoT devices and the total network throughput. The network that uses the proposed DIF-PA and the proposed CG-APA can provide the required service to an average of $5.9$ IoT devices, which is better than the network that uses the proposed DIF-PA and Nearest-APA around $8-12$\%. Without the proposed DIF-PA, the system that uses Equal-PA with CG-APA only serves $3.5$ IoT devices when the network exists $12$ IoT devices, and this number decreases quickly to $1.5$ when the number of IoT devices increases to $15$. This situation is because as more IoT devices join the network, the mutual interference will substantially increase as well. Nonetheless, Equal-PA does not adjust the power levels appropriately, resulting in a significant increase in the gap between the proposed DIF-PA and Equal-PA. With the modified BB algorithm, the system maintains the number of satisfied IoT devices at $7.8$ with a total network throughput of $3.9$ (bits/s/Hz) since this method keeps adding IoT devices to the system one after another without considering how many IoT devices are joining the network. However, the total network throughput when the system applies the proposed DIF-PA and CG-APA increases from $6.9$ to $7.5$ (bits/s/Hz) when the number of IoT devices increases from $12$ to $18$.

    \begin{table}[t] 
		\caption{Number of satisfied IoT devices and the running time (millisecond) with different network sizes. The left and right values show the number of satisfied IoT devices, while the right value indicates the running time.}
		\centering
        \resizebox{\columnwidth}{!}{\begin{tabular}{|c|c|c|c|c|c|} 
			\toprule[1pt]\midrule[0.3pt] 
			K&N & $^\text{DIF-PA}_\text{+ CG-APA}$ & $^\text{DIF-PA}_\text{+ Nearest-APA}$ & \text{BB Algorithm} & $\text{Brute Force}$\\ 
			\midrule 
			2 & 8 & $3.7/0.97$ & $3.47/0.94$ & $5.03/1.2$ & $5.2/1.8$ \\ 
			\midrule 
			3 & 8 & $4.6/1.3$ & $4.3/1.1$ & $6/2.6$ & $6.5/64.1$ \\ 
			\midrule 
			3 & 12 & $5.8/1.5$ & $5.5/1.3$ & $6.5/2.7$ & $7.4/91141$ \\ 
			\midrule 
			5 & 15 & $6.3/2.3$ & $5.9/2.1$ & $7.8/41$ & $-/-$ \\ 
			\midrule 
			5 & 20 & $6.5/4.9$ & $6.2/4.4$ & $7.8/41.4$ & $-/-$ \\ 
			\midrule 
			5 & 25 & $6.7/9.5$ & $6.4/8.7$ & $7.8/41.9$ & $-/-$ \\ 
			% \midrule 
			% 5 & 30 & $6.8/16.2$ & $6.4/15.8$ & $7.8/42$ & $-/-$ \\ 
			% \midrule 
			% 5 & 35 & $6.9/27.9$ & $6.5/26.5$ & $7.8/43.65$ & $-/-$ \\ 
			% \midrule 
			% 5 & 45 & $6.9/60.9$ & $6.5/58.4$ & $7.8/46$ & $-/-$ \\ 
			\bottomrule[1pt] 
		\end{tabular}}
		\label{Table: Computation Time}
		\vspace{-5pt}
	\end{table}%
    Table \ref{Table: Computation Time} shows the system's performance in terms of the number of satisfied IoT devices and the running time when the required data rate of all IoT devices is $0.5$ (bits/s/Hz). The system's running time that applies Equal-PA and Nearest-APA can be ignored due to the lack of computational requirements. Meanwhile, the running time of the proposed CG-APA can be evaluated by observing its combination with the proposed DIF-PA. Thus, the combination of Equal-PA and the proposed CG-APA will not be illustrated. For convenience, we call the combination of the proposed DIF-PA and the CG-APA as the proposed method. We can observe that the brute force method always achieves the highest value because it searches all candidates. The number of IoT devices the system can serve when applying the brute force method is higher than those of the proposed method and the BB algorithm, approximately $20$\% and $8$\%. However, the computational cost makes this gap impractical. The running time of the brute force method increases dramatically. For instance, with $2$ APs and $8$ IoT devices, the running time is only around $2$ms and exponentially increases to $64$ms when the system has $3$ APs. The running time becomes unacceptable, around $91,141$ms, when the system has $3$ APs and $12$ IoT devices. Meanwhile, with the same setup, the running times of the proposed method and the BB algorithm are only $1.7$ms and $2.7$ms, respectively. As the network expands, the brute force method's running time is too long, leading us to disregard it in further comparisons. In comparison between the proposed method and the BB algorithm,  both methods request more running time when more IoT devices join the networks. However, the proposed method is significantly more time-efficient. For example, with $5$ APs, the running time increases from $2.3$ms for 15 IoT devices to $9.5$ms for 25 devices. When the system uses Nearest-APA combined with DIF-PA, the running time is reduced by about $5$\% compared to the proposed method. Meanwhile, the running time of the BB algorithm highly depends on the number of APs since the number of nodes that the BB tree must search at each level increases exponentially with the number of APs. The running time of the BB algorithm is $2.7$ms with $12$ IoT devices and $3$ APs in the network and increases to $41$ms with $15$ IoT devices and $5$ APs in the network.
    %  and to $60.9$ms for $45$ devices
    
    \begin{figure}
    		\includegraphics[trim=3.6cm 8.3cm 3.5cm 9.0cm, clip=true, width=3.6in]{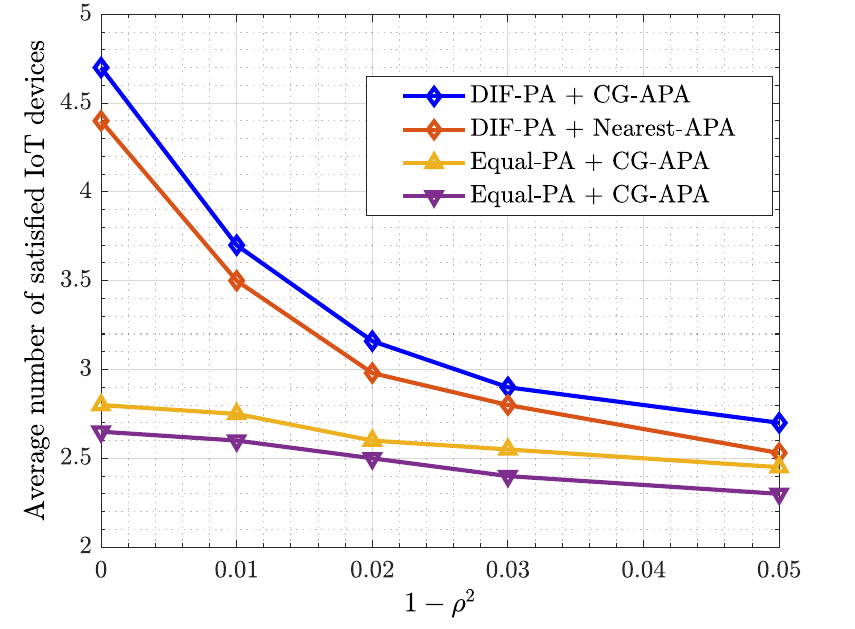}
    		\caption{Number of satisfied IoT devices when the CSI is imperfect.}
    		\label{fig: ImperfectChannelEvaluation}
            \vspace{-1.5em}
    \end{figure} 

    \begin{table*}[t] 
		\caption{Number of satisfied IoT devices (Num-sat-IDs), the total network throughput (Tot-throughput), and the running time (second) with large network sizes.}
		\centering
		\resizebox{\textwidth}{!}{\begin{tabular}{|c|c|c|c|c|c|c|c|c|c|c|} 
			\toprule[1pt]\midrule[0.3pt] 
			K&N & \multicolumn{3}{|c|}{${\mbox{DIF-PA + CG-APA}}$} & \multicolumn{3}{|c|}{${\mbox{DIF-PA + Nearest-APA}}$} & \multicolumn{3}{|c|}{${\mbox{Modified BB Algorithm}}$} \\ 
            \cline{3-11}
            & & $\mbox{Num-sat-IDs}$ & $\mbox{Tot-throughput}$ & $\mbox{Running time}$ & $\mbox{Num-sat-IDs}$ & $\mbox{Tot-throughput}$ & $\mbox{Running time}$ & $\mbox{Num-sat-IDs}$ & $\mbox{Tot-throughput}$ & $\mbox{Running time}$ \\
			\midrule 
			20&50 & $14.4$ & $41.4$ & $0.07$ & $10.8$ & $30.1$ & $0.06$ & $19.74$ & $19.74$ & $0.7$  \\ 
			\midrule 
			30&100 & $22.1$ & $57.9$ & $0.61$ & $17.8$ & $44.8$ & $0.43$ & $29.3$ & $29.3$ & $6.2$  \\ 
			\midrule 
			50&100 & $28.7$ & $81.4$ & $0.67$ & $22.7$ & $67.9$ & $0.49$ & $47.1$ & $47.1$ & $21.1$  \\ 
			\midrule 
			50&200 & $35.2$ & $93.3$ & $5.72$ & $27.7$ & $71.9$ & $4.63$ & $47.2$ & $47.2$ & $56.9$  \\ 
			\midrule 
			70&200 & $41.7$ & $114.1$ & $6.75$ & $35.5$ & $97.9$ & $4.82$ & $67.9$ & $67.9$ & $146.9$  \\ 
			\bottomrule[1pt] 
		\end{tabular}}
		\label{Table: Computation Time large network sizes}
		% \vspace{-5pt}
	\end{table*}%
    
    To evaluate the performance of the system under imperfect CSI, we consider the model where the estimated channel between the $n$-th IoT device and the $k$-th AP is represented by $h^{'}_{k,n} = \rho h_{k,n}$ where $\rho $ is the imperfect factor, and the channel error is calculated as $\sqrt{1-\rho^2}\mathcal{CN}(0,\vartheta_{k,n})$. Accordingly, we can calculate the SINR of $n$-th IoT device as $\gamma_{n}(\boldsymbol{\Psi},\textbf{\textrm{p}})=\frac{ \sum\limits_{k\in\mathcal{K}}\rho^2|\mu_{k,n}h_{k,n}|^2 P_{n}}{\sum\limits_{n'\neq n, n'\in\mathcal{N}}\sum\limits_{k\in\mathcal{K}}|\mu_{k,n'}h_{k,n}|^2 P_{n'}+\sum\limits_{k\in\mathcal{K}}\mu_{k,n}(1-\rho^2)\vartheta_{k,n}+\sigma^2}$. All proposed approaches to perfect CSI circumstances remain unchanged for imperfect CSI. Fig.~\ref{fig: ImperfectChannelEvaluation} represents the performance of the proposed methods versus the error factor ${1-\rho^2}$. 
    
    We can observe that the performance of all methods decreases when the error factor increases due to the fact that a higher error factor leads to lower received signal strength and higher interference. The performance of the proposed DIF-PA decreases more rapidly compared to Equal-PA under imperfect CSI. This is because the SINR is impacted by the lower signal strength and increased interference due to channel errors. Additionally, fixing the data rate for users, who often have good channels, makes DIF-PA more vulnerable to imperfection, accelerating performance reduction. Conversely, with equal power allocation, the number of satisfied IoT devices decreases linearly as the additional interference from channel estimation error is fixed. Specifically, the number of satisfied IoT devices using the proposed DIF-PA and CG-APA reduces significantly from $4.7$ with ${1-\rho^2} = 0$, the perfect CSI scenario, to $2.7$ with ${1-\rho^2} = 0.05$. Meanwhile, the CG-APA method still outperforms the nearest-APA method by $5$-$10$\%, even under imperfect CSI conditions.

    % {\color{blue}The use of perfect CSI and other ideal assumptions in this study provides a clear path to understanding the potential of our optimization framework. However, we recognize that in practical IoT networks, these conditions are rarely met. Future work will focus on extending this framework to account for imperfect CSI, larger and more dynamic networks, and other real-world challenges. Preliminary simulations under less ideal conditions suggest that our approach remains robust, though further refinement is necessary.}

    Table.~\ref{Table: Computation Time large network sizes} provides a comparative analysis of three different methods, including DIF-PA + CG-APA, DIF-PA + Nearest-APA, and the modified BB algorithm, across large network sizes. The table presents key metrics, including the number of satisfied IoT devices (Num-sat-IDs), total network throughput (Tot-throughput), and the running time required by each method. The results indicate that the combination of the DIF-PA and CG-APA methods delivers a strong balance between the number of satisfied IoT devices and total throughput, consistently outperforming the system deploying the DIF-PA and Nearest-APA methods. While the modified BB algorithm achieves slightly higher satisfaction in terms of IoT devices, the system that uses the proposed DIF-PA and CG-APA methods shows a significantly more efficient running time. This efficiency becomes particularly apparent as the network size increases, making the proposed method more scalable than the modified BB algorithm. For example, with $50$ APs and $100$ IoT devices, the DIF-PA and CG-APA method requires only $0.67$ seconds, compared to $21.1$ seconds for the modified BB method. As the number of IoT devices increases to $200$, the running time of all methods grows, but the DIF-PA and CG-APA combination remains efficient, with a running time of $5.72$ seconds for $50$ APs and only increasing slightly to $6.75$ seconds for $70$ APs. In contrast, the modified BB method requires $146.9$ seconds for $70$ APs, illustrating a significant difference in scalability and computational efficiency. The DIF-PA combined with the Nearest-APA method, while providing lower total network throughput and fewer satisfied IoT devices, demonstrates a notably low running time, approximately $20$\% lower than the DIF-PA and CG-APA combination. This trade-off demonstrates the efficiency of the resource-sharing scheme, suggesting that further accelerated methods for the GA can be applied to improve running time while maintaining or even enhancing overall network performance. Note that all simulation results were obtained using a personal computer. In practice, the computation can be significantly faster on a higher-performance computer or cloud infrastructure, reducing the running time.
    
    We recognize that while our proposal has demonstrated the efficiency and strong potential of the resource-sharing scheme, further refinement will be conducted to enhance its practical applicability. Specifically, ongoing research will focus on improving resource-sharing strategies to better adapt to imperfect CSI and large-size networks. These enhancements will ensure that our framework remains robust and effective across a broader range of real-world IoT environments.
    
    \section{Conclusion}\label{Conclusion}
    We investigated the infeasible issue in IoT networks where the QoS requirements cannot simultaneously be satisfied for all the IoT devices. %Specifically, an IoT device is considered to be satisfied if the IoT device meets the minimum data throughput.
    To make the network more practical and robust, we identified the satisfied IoT devices that the system is capable of providing.  Accordingly, with respect to the device association and power allocation, we proposed to solve the two problems, i.e., the number of satisfied IoT devices maximization and the dual-objective problem, including the number of satisfied IoT devices and the total network throughput. The former problem was solved via a modified BB algorithm. Moreover, using the modified BB algorithm, the infeasibility of the system was determined. The dual-objectives problem was then solved with a higher priority for the number of satisfied IoT devices.
    Since the formulated problem is non-convex, an iterative algorithm gradually increases the number of satisfied IoT devices by sharing the power budget from the high data throughput to lower ones.  The device association was attained by applying the coalition game, while a dual fixed-point algorithm was provided to obtain the power allocation. Simulation results demonstrated the efficiency of both solving the maximization of the cardinality of the satisfied set and solving the dual-objective problem. The modified BB algorithm could provide service to more IoT devices but with lower total network throughput and longer running time. In addition, when the system is in extremely harsh conditions, the system applying the proposed DIF-PA was more efficient in both the number of satisfied IoT devices and the total network throughput compared to other methods. In comparison to Equal-PA, the system that uses the proposed DIF-PA serves more IoT devices of up to $50$\%. Our proposal has shown strong resource-sharing efficiency, but there is room for further enhancement. To improve performance under imperfect CSI and in large-scale networks, ongoing research will refine our strategies, ensuring the framework remains robust and adaptable to real-world IoT environments.
    
	\appendices
 	\section{Proof of Lemma \ref{LogLemma_NoRateConstraint}} \label{ProofLemmaLogApproximation}
\textcolor{black}{Let us first denote the optimal solution to problem~\eqref{MaxRwithoutRconstraint_v2} in the $i$-th iteration of Algorithm~\ref{LogApproximationAlgorithm_EqRateConstraint} as $\textbf{\textrm{p}}^*[i]$. Then, we observe the following series of inequalities} 
% 	{\color{blue}I will remove redundant parts after getting your confirmation for this proof}
	\begin{equation}
		\begin{aligned}
			{R}_{\textrm{tot}}(\textbf{\textrm{p}}^*[i-1]) &\stackrel{(a)}{=} \tilde{R}_{\textrm{tot}}(\textbf{\textrm{p}}^*[i-1],\boldsymbol{\alpha}(\textbf{\textrm{p}}^*[i-1]),\boldsymbol{\beta}(\textbf{\textrm{p}}^*[i-1])\\ &\stackrel{(b)}{\leq} \tilde{R}_{\textrm{tot}}(\textbf{\textrm{p}}^*[i],\boldsymbol{\alpha}(\textbf{\textrm{p}}^*[i-1]),\boldsymbol{\beta}(\textbf{\textrm{p}}^*[i-1])\\
            &\stackrel{(c)}{\leq} R_{\textrm{tot}}(\textbf{\textrm{p}}^*[i]).
% 			\\
% 			& \stackrel{(c)}{\leq} R(\textbf{\textrm{p}}^{(i)^*}) \stackrel{(a)}{=} \tilde{R}(\textbf{\textrm{p}}^{(i)^*},\boldsymbol{\alpha}(\textbf{\textrm{p}}^{(i)^*}),\boldsymbol{\beta}(\textbf{\textrm{p}}^{(i)^*})).
		\end{aligned}
	\end{equation}
	At $\textbf{\textrm{p}}^*[i-1]$, we have ${R}_{\textrm{tot}}(\textbf{\textrm{p}}^*[i-1]) = \tilde{R}(\textbf{\textrm{p}}^*[i-1])$ since $\boldsymbol{\alpha}$ and $\boldsymbol{\beta}$ are calculated at $\textbf{\textrm{p}}^*[i-1]$ as in \eqref{LogApproximation1}. Next, we optimally solve \eqref{MaxRwithoutEqualRconstraint_v4} to obtain $\textbf{\textrm{p}}^*[i-1])$, and thus $(b)$ is always true. From \cite{Singh2015}, we have $\log_2(1+z)\geq \alpha \log(z) + \beta$ for any $\alpha$ and $\beta$. Therefore, $\log_2(1+\gamma_{n,l}(\textbf{\textrm{p}}^*[i]))\geq \alpha_{n,l}(\textbf{\textrm{p}}^*[i-1]) \log_2(\gamma^{(i)}_{n,l}(\textrm{\textbf{p}}^*[i])) + \beta_{n,l}(\textbf{\textrm{p}}^*[i-1])$. Consequently, ${R}_{\textrm{tot}}(\textbf{\textrm{p}}^*[i])\geq \tilde{R}(\textbf{\textrm{p}}^*[i])$, and the inequatlity $(c)$ is satisfied. In addition, $R_{\textrm{tot}}(\textbf{\textrm{p}})$ is bounded above due to the limited power budget. As a result, solving the problem \eqref{MaxRwithoutRconstraint_v2} using Algorithm \ref{LogApproximationAlgorithm_EqRateConstraint}, the value of $R_{\textrm{tot}}(\textbf{\textrm{p}})$ always increases and converges. The proof is completed.

    % At the converged point after $I$ steps, the log approximation factors are calculated at $\textbf{\textrm{p}}^*[I]$ as in \eqref{LogApproximation1}. Thus, the lower bound $\tilde{R}_{n}(\textbf{\textrm{p}})$ is exactly the same as the actual data throughput \eqref{RateCalculation} at this point. 
    
    \section{Proof of Lemma \ref{FixedPointConvergence_noRateConstraint}} \label{Proof:FixedPointConvergence_noRateConstraint}
    \begin{figure*}[t]
        \begin{equation}\label{DerivativeIF}
            \begin{aligned}
                &\frac{\partial I_n(\textbf{\textrm{p}})}{\partial P_n}=\frac{\left(B\alpha_{n,l}\log_{2}(\textrm{e})\right)\underset{{j\neq n, j\in \mathcal{N}}}{\sum}B\alpha_{j,l}\frac{\ln (2) \left(|\boldsymbol{\mu}_j^T\boldsymbol{g}_{n}|^2\right)^2}{\left(\ln (2) \left(\underset{{n'\neq j,n'\in \mathcal{N}}}{\sum}|\boldsymbol{\mu}_{j}^T \boldsymbol{g}_{n'}|^2 P_{n'} + \sigma^2\right)\right)^2}}{\left(\underset{n'\neq n}{\sum} B\alpha_{n'}({\textbf{\textrm{p}}}[i-1])\frac{\sum_{k=1}^K |\mu_{k,n} g_{k,n'}|^2}{\underset{j\neq n'}{\sum}\sum_{k=1}^K|\mu_{k,j}g_{k,n'}|^2P_j+\sigma^2}+\sum_{k=1}^{K}\ln(2)\theta_{k}\mu_{k,n}\right)^2} >0, \forall n\in \mathcal{N}/\mathcal{Q} .\\
                &\frac{\partial I_n(\textbf{\textrm{p}})}{\partial P_n}=0, \forall n\in \mathcal{Q}.
            \end{aligned}
        \end{equation}
        \hrule
        \vspace*{-0.25cm}
    \end{figure*}
    From \cite[Theorem~2]{Yates1995}, three standard interference function conditions include
    \begin{enumerate}
        \item Positivity: $\boldsymbol{I}(\normalfont{\textbf{\textrm{p}}})>0$.
        \item Monotonicity: If $\normalfont{\textbf{\textrm{p}}}\succeq\normalfont{\textbf{\textrm{p}}}'$, then  $\boldsymbol{I}(\normalfont{\textbf{\textrm{p}}})\succeq\boldsymbol{I}(\normalfont{\textbf{\textrm{p}}}')$.
        \item Scalability: For any $\gamma >1$, $\gamma \boldsymbol{I}(\normalfont{\textbf{\textrm{p}}})\succ\boldsymbol{I}(\gamma \normalfont{\textbf{\textrm{p}}})$.
    \end{enumerate}
    We first observe from \eqref{FixedPointEq1:UnsatisfiedIoT devices} and \eqref{FixedPointEq1:SatisfiedIoT devices} that the positivity is always guaranteed due to the non-negative values of power. Next, we take the derivative $\frac{\partial I_n(\textbf{\textrm{p}})}{\partial P_n}>0$ as in \eqref{DerivativeIF}. The derivative of $I_n(\textbf{\textrm{p}}), \forall n\in \mathcal{N}$ is always positive. Therefore, $I_n(\textbf{\textrm
    p}), \forall n\in \mathcal{N}$ is an increasing function. We can deduce that $I_n(\textbf{\textrm{p}}_1) > I_n(\textbf{\textrm{p}}_2)$ if $\textbf{\textrm{p}}_1\succeq\textbf{\textrm{p}}_2$ and the monotonicity property is ensured.
     Finally, for any real number $\gamma>1$, we can easily show that $\gamma I_n(\textbf{\textrm{p}})>I_n(\gamma\boldsymbol{ P}), \forall n\in \mathcal{N}$. Thus, the fixed-point equation in \eqref{FixedPointWithEqRateConstraint} is a standard interference function. Consequently, it converges into a unique solution. Besides, \eqref{FixedPointEq2} is desired from the conditions to optimize the Lagrangian function and provides a sub-optimal solution to the problem in \eqref{MaxRwithEqualRconstraint_v5}. Furthermore, $P_n=\exp(\bar{P}_n)$ is a one-to-one transformation, and thus the convergence point of \eqref{FixedPointEq2} maximize the objective of \eqref{MaxRwithEqualRconstraint_v5}. The proof is completed.

    \section{Proof of Lemma \ref{Lemma: Convergence}}
    \label{Proof: Convergence}
    Given the initial set of the AP association $\boldsymbol{\Psi}^{*,(0)}$, solving the problem in \eqref{MaxRwithoutRconstraint_v2} to achieve $\textbf{\textrm{p}}^{*,(0)}$ and determine the satisfied set at $t=0$, $\mathcal{Q}^{*}(\boldsymbol{\Psi}^{*,(0)},\textbf{\textrm{p}}^{*,(0)})$. Given fixed AP-IoT devices association, Algorithm \ref{LogApproximationAlgorithm_EqRateConstraint} is used to solve the problem in \eqref{MaxRwithoutRconstraint_P_v1} and yields the solution $\textbf{\textrm{p}}^{*,(t)}$. As the constraint \eqref{MaxRwithoutRconstraint_v1:EqualRConstraint} is considered in \eqref{MaxRwithoutRconstraint_P_v1}, we then have 
    	\begin{equation}\label{ConvergenceProof1}
    			|\mathcal{Q}(\boldsymbol{\Psi}^{*,(t-1)},\textbf{\textrm{p}}^{*,(t)})|\geq |\mathcal{Q}(\boldsymbol{\Psi}^{*,(t-1)},\textbf{\textrm{p}}^{*,(t-1)})|.
    	\end{equation}
%        \begin{equation}\label{ConvergenceProof1}
%            \begin{cases}
%                &|\mathcal{Q}(\boldsymbol{\Psi}^{*,(t-1)},\textbf{\textrm{p}}^{*,(t)})|\geq |\mathcal{Q}(\boldsymbol{\Psi}^{*,(t-1)},\textbf{\textrm{p}}^{*,(t-1)})|\\
%                &R_{\textrm{tot}}(\boldsymbol{\Psi}^{*,(t-1)},\textbf{\textrm{p}}^{*,(t)})\geq R_{\textrm{tot}}(\boldsymbol{\Psi}^{*,(t-1)},\textbf{\textrm{p}}^{*,(t-1)})
%            \end{cases}
%        \end{equation}
        After solving the power allocation, the AP-IoT devices association is obtained via Algorithm \ref{AlgorithmCG}. Therein, the updating conditions \eqref{CGcondition1} and \eqref{CGcondition2} guarantees 
        \begin{equation}\label{ConvergenceProof2}
            \begin{cases}
                &|\mathcal{Q}(\boldsymbol{\Psi}^{*,(t)},\textbf{\textrm{p}}^{*,(t)})|\geq |\mathcal{Q}(\boldsymbol{\Psi}^{*,(t-1)},\textbf{\textrm{p}}^{*,(t)})|,\\
                &R_{\textrm{tot}}(\boldsymbol{\Psi}^{*,(t)},\textbf{\textrm{p}}^{*,(t)})\geq R_{\textrm{tot}}(\boldsymbol{\Psi}^{*,(t-1)},\textbf{\textrm{p}}^{*,(t)}).
            \end{cases}
        \end{equation}
        Based on \eqref{ConvergenceProof1} and \eqref{ConvergenceProof2}, we have 
        \begin{equation}\label{ConvergenceProof3}
        	|\mathcal{Q}(\boldsymbol{\Psi}^{*,(t)},\textbf{\textrm{p}}^{*,(t)})|\geq |\mathcal{Q}(\boldsymbol{\Psi}^{*,(t-1)},\textbf{\textrm{p}}^{*,(t-1)})|.
        \end{equation}
%        \begin{equation}\label{ConvergenceProof3}
%            \begin{cases}
%                &|\mathcal{Q}(\boldsymbol{\Psi}^{*,(t)},\textbf{\textrm{p}}^{*,(t)})|\geq |\mathcal{Q}(\boldsymbol{\Psi}^{*,(t-1)},\textbf{\textrm{p}}^{*,(t-1)})|\\
%                &R_{\textrm{tot}}(\boldsymbol{\Psi}^{*,(t)},\textbf{\textrm{p}}^{*,(t)})\geq R_{\textrm{tot}}(\boldsymbol{\Psi}^{*,(t-1)},\textbf{\textrm{p}}^{*,(t-1)})
%            \end{cases}
%        \end{equation}
        We now guarantee a non-decrease in the number of satisfied IoT devices. When the number of satisfied IoT devices converges, the power allocation and the AP-IoT devices association are still optimized until the total network throughput converges also. Obtaining the power allocation via DIF-PA algorithm ensures the improvement in the total network throughput. Based on \eqref{ConvergenceProof2}, we can ensure the non-decreasing of the total network throughput when the satisfied set is fixed. 
        In addition, all variables are limited by constraints, and thus both the number of satisfied IoT devices and the total network throughput value are bounded. Therefore, the convergence of Algorithm \ref{AAlgorithm} is guaranteed. Besides, due to the consideration of the infeasible problem, the local or global solution is not taken into account, and Algorithm \ref{AAlgorithm} provides a good solution to the problem. The proof is now completed.
    \ifCLASSOPTIONcaptionsoff
    \newpage
    \fi
    
    \bibliographystyle{IEEEtran}
    \bibliography{Bib1}
\end{document}